\documentclass{llncs}

\usepackage[english]{babel}
\usepackage{tikz}
\usetikzlibrary{arrows.meta,decorations.pathmorphing,positioning}

\usepackage{relsize}
\usepackage{caption}
\usepackage{algorithm}
\usepackage{algpseudocode}

\usepackage{amsmath}
\usepackage{graphicx}
\usepackage[colorlinks=true, allcolors=blue, bookmarks=true]{hyperref}
\usepackage{bookmark}

\usepackage{mdframed}
\usepackage{amsfonts}
\usepackage{amssymb}
\usepackage{stmaryrd}
\usepackage{dsfont}
\usepackage{xcolor}

\usepackage{xspace}
\usepackage{bbm}
\usepackage[many,theorems]{tcolorbox}

\usepackage{float}

\usepackage{graphicx} %

\usepackage{enumitem}
\usepackage{microtype}
\usepackage{cleveref}
\usepackage[asymptotics,operators,probability,sets]{cryptocode}

\newtcbtheorem[
  crefname={algorithm}{algorithms},
  Crefname={Algorithm}{Algorithms}
  ]{algobox}{Algorithm}{%
  breakable,
  lower separated=false,
  colframe=white!20!black,
  fonttitle=\bfseries,
  colbacktitle=gray!120,
  coltitle=white,
  enhanced,
  colback=white,
  size=small,
  boxrule=0.4pt,
  titlerule=1pt,
  sharp corners,
}{algo}

\floatname{algorithm}{Procedure}
\crefname{algorithm}{procedure}{procedures}
\Crefname{algorithm}{Procedure}{Procedures}

\crefname{figure}{figure}{figures}
\Crefname{figure}{Figure}{Figures}

\newcommand{\vecfont}{\mathbf}
\newcommand{\probfont}{\text}
\newcommand{\qgatefont}{\mathrm}
\newcommand{\algofont}{\mathsf}

\newcommand{\Am}{\vecfont{A}}
\newcommand{\Hm}{\vecfont{H}}

\newcommand{\xv}{\vecfont{x}}
\newcommand{\yv}{\vecfont{y}}

\newcommand{\sv}{\vecfont{s}}
\newcommand{\av}{\vecfont{a}}

\newcommand{\vv}{\vecfont{v}}

\newcommand{\ev}{\vecfont{e}}
\newcommand{\uv}{\vecfont{u}}
\newcommand{\tv}{\vecfont{t}}

\newcommand{\zv}{\vecfont{z}}
\newcommand{\rv}{\vecfont{r}}

\newcommand{\Zv}{\vecfont{Z}}
\newcommand{\Bm}{\vecfont{B}}
\newcommand{\Zm}{\vecfont{Z}}
\newcommand{\zerov}{\vecfont{0}}

\newcommand{\hf}{\widehat{f}}

\newcommand{\tAm}{\widetilde{\Am}}

\newcommand{\ket}[1]{|#1\rangle}
\newcommand{\bra}[1]{\langle#1|}
\newcommand{\ketbra}[2]{|#1\rangle\langle#2|}
\newcommand{\braket}[2]{\langle #1 | #2 \rangle}
\newcommand{\altketbra}[1]{\ketbra{#1}{#1}}
\newcommand{\kb}[1]{\altketbra{#1}}

\newcommand{\QFTw}{\qgatefont{QFT}}
\newcommand{\QFT}{\QFTw}
\newcommand{\Ugate}{\mathsf{U}}

\renewcommand{\norm}[1]{\ensuremath{\lvert  \kern-1pt\lvert #1 \rvert \kern-1pt \rvert}}
\newcommand{\F}{\mathbb{F}}

\newcommand{\Z}{\mathbb{Z}}
\newcommand{\zo}{\{0,1\}}
\newcommand{\setmid}{\ :\ }
\newcommand{\Iint}[2]{\llbracket #1 , #2 \rrbracket}

\newcommand{\SIS}{\ensuremath{{\probfont{SIS}}}}
\newcommand{\LWE}{\ensuremath{{\probfont{LWE}}}}

\newcommand{\SLWE}{\ensuremath{S\ket{\LWE}}}
\newcommand{\CLWE}{\ensuremath{C\ket{\LWE}}}
\newcommand{\ICLWE}{\ensuremath{IC\ket{\LWE}}}
\newcommand{\ISIS}{\ensuremath{\probfont{ISIS}}}

\newcommand{\ICC}{\ISIS}

\newcommand{\EDCP}{\ensuremath{{\probfont{EDCP}}}}
\newcommand{\EDCPb}{\ensuremath{{\overline{\EDCP}}}}

\renewcommand{\aa}{\algofont{A}}
\newcommand{\rdtape}{\mathcal{R}}
\newcommand{\eqdef}{\triangleq}
\renewcommand{\poly}{\pcpolynomialstyle{poly}}
\renewcommand{\negl}{\pcpolynomialstyle{negl}}
\newcommand{\Time}{\mathrm{Time}}

\renewcommand{\sample}{\gets}
\newcommand{\Unif}{\sample}
\newcommand{\E}{\mathbb{E}}

\newcommand{\ie}{\textit{i.e.}\xspace}

\newcommand{\andre}[1]{\textcolor{blue}{[A : #1]}}
\newcommand{\paul}[1]{\textcolor{red}{[P : #1]}}

\newcommand{\isanonymous}{0}

\newcommand{\COMMENT}[1]{}
\newcommand{\eps}{\varepsilon}
\newcommand{\one}{\mathbbm{1}}
\newcommand{\T}{^{\mathsf{T}}}

\newcommand{\horzbar}{\rule[.5ex]{2.5ex}{0.5pt}}

\newcommand{\wpsi}{\widetilde{\psi}}
\newcommand{\hzsy}{z_{\sv,\xv}}
\newcommand{\ohzsy}{\overline{z_{\sv,\xv}}}
\newcommand{\C}{\Lambda}
\newcommand{\Pdec}{p}
\newcommand{\wphi}{{W'}}

\newcommand{\rdext}{\algofont{E}}
\newcommand{\xf}{\overline{\xv}}
\newcommand{\rdtapef}{\overline{\rdtape}}
\newcommand{\isissolver}{\algofont{ISIS\text{-}Solver}}
\newcommand{\matext}{\algofont{MatExtend}}
 
\title{On the Quantum Equivalence between $\SLWE$ and $\ISIS$}

\if \isanonymous 0 {
	\author{André Chailloux \and Paul Hermouet}
	\institute{Inria de Paris, COSMIQ team\\\email{andre.chailloux@inria.fr}\hfil\email{paul.hermouet@inria.fr}}
} \else {
	\author{}
	\institute{}
}
\fi

\makeatletter%
\begin{document}
\maketitle

\begin{abstract}
	\normalsize
	Chen, Liu, and Zhandry~\cite{CLZ22} introduced the problems $\SLWE$ and $\CLWE$ as quantum analogues of the Learning with Errors problem, designed to construct quantum algorithms for the Inhomogeneous Short Integer Solution (ISIS) problem. 
	Several later works have used this framework for constructing new quantum algorithms in specific cases. However, the general relation between all these problems is still unknown.

	In this paper, we investigate the equivalence between $\SLWE$ and $\ISIS$. We present the first fully generic reduction from ISIS to $\SLWE$, valid even in the presence of errors in the underlying algorithms. We then explore the reverse direction, introducing an inhomogeneous variant of $\CLWE$, denoted $\ICLWE$, and show that $\ICLWE$ reduces to $\SLWE$. Finally, we prove that, under certain recoverability conditions, an algorithm for ISIS can be transformed into one for $\SLWE$. We instantiate this reverse reduction by tweaking a known algorithm for (I)$\SIS_\infty$ in order to construct a quantum algorithm for $\SLWE$ when the alphabet size $q$ is a small power of $2$, recovering some results of Bai et al.~\cite{BJK+25}. Our results thus clarify the landscape of reductions between $\SLWE$ and $\ISIS$, and we show both their strong connection and the remaining barriers for showing full equivalence.

\end{abstract} 
\newcommand{\SEDCP}{\overline{\text{EDCP}}}

\section{Introduction}
\label{sec:intro}
\subsection{Context}
A cornerstone of lattice-based cryptography is Regev's reduction~\cite{Reg05}, which is a quantum reduction between some lattice-based problems related to the \emph{Short Integer Solution} (SIS) problem and the \emph{Learning With Errors} (LWE) problem. This is a quantum reduction that uses a (classical or quantum) algorithm for LWE in order to create a superposition of noisy lattice points and then measuring in the Fourier basis to obtain a short dual lattice point. 

As noted in~\cite{SSTX09}, we actually need to solve an easier problem than $\LWE$, since the error can be in quantum superposition. This creates strong links between this problem and the \emph{Dihedral Coset Problem}. This was actually first explicitly used by Brakerski, Kirshanova, Stehlé and Wen~\cite{BKSW18} where they extend this reduction and introduce the Extrapolated Dihedral Coset Problem.

A few years later, Chen Liu and Zhandry~\cite{CLZ22} revisited this reduction for algorithmic purposes. They show that in some regimes, the $\LWE$ problem with errors in quantum superposition, which they call the $\SLWE$ problem, can be significantly easier than the standard $\LWE$ problem. They show how to construct polynomial time quantum algorithm for the $\SIS_\infty$ for some contrived parameters. At that time, we didn't have corresponding classical algorithms but later works showed efficient classical algorithms for this problem~\cite{II24,KOW25}.

 Using this type of algorithms, Yamakawa and Zhandry~\cite{YZ24} provided a first quantum advantage without structure in the Random Oracle Model. They showed that this approach cannot in all generality be ``dequantized" as was the previously mentioned algorithm for $\SIS_\infty$. This reduction has also been extended to the setting of linear codes~\cite{DRT23,CT24}, as well as structured codes in order to obtain a quantum advantage~\cite{JSW+24,CT25}.

All of these results use an algorithm for $\SLWE$ to construct an algorithm for $\SIS$ or $\ISIS$ and perform an ad hoc analysis of this reduction.  Note that we cannot have a generic reduction from $\SIS$ to $\SLWE$ using this approach (see~\cite{CT24}). Chailloux and Tillich~\cite{CT25} provided the first reduction, from $\ISIS$ to $\SLWE$. This reduction however requires some assumptions on the $\SLWE$ algorithm which are satisfied by classical algorithms but not necessarily by quantum algorithms. A first natural question arises.

\begin{question} Is it possible to have a fully general reduction from $\ISIS$ to $\SLWE$ that is robust to errors in the $\SLWE$ algorithm?
\end{question}
Also, because of the importance of $\SLWE$, recent works directly construct quantum algorithms for $\SLWE$. First, a generic quantum algorithm for $\SLWE$ was presented in ~\cite{CHL+25}, running in subexponential time and requiring a subexponential number of queries. Then, the authors of~\cite{BJK+25} presented a slightly superpolynomial algorithm for $\SLWE$ in the case where the alphabet size $q$ is a small power of $2$. These results both use variants of the quantum Kuperberg sieve~\cite{Kup13} for the Dihedral Coset Problem. 
When looking at these algorithms more carefully, one can notice that they are actually very similar to known classical algorithm for $\ISIS$ and this raises the following natural question
\begin{question} Is there a way to construct algorithms for $\SLWE$ from algorithms for $\ISIS$? More generally, are the problems $\SLWE$ and $\ISIS$ equivalent?
\end{question}

\noindent \textbf{Contributions in short.} Our work has the following contributions:
\begin{enumerate}
	\item We provide the first generic reduction from $\ISIS$ to $\SLWE$ without any condition on the $\SLWE$ algorithm, even when it has errors. This improves and generalizes the work of~\cite{CT25}. This result is obtained by a slight change in the construction of the $\ISIS$ algorithm, as well as a much sharper analysis of the errors in the $\SLWE$ algorithm.
	\item For the reverse reduction, we show that, given an $\ISIS$ algorithm satisfying certain conditions, we can construct an algorithm for the corresponding $\SLWE$ problem, giving a conditional reduction. We require that the output of the $\ISIS$ algorithm is uniform over the solutions and that one can recover the internal randomness used from a solution. This reverse reduction is done by introducing an intermediate problem called $\ICLWE$, which is the inhomogeneous of the $\CLWE$ problem introduced in~\cite{CLZ22}.
	\item The condition on the $\ISIS$ algorithm seems quite stringent at first sight and a natural question to ask is whether it can be useful at all. We answer positively, giving examples of $\ISIS$ algorithms that satisfy this property. Moreover, we show that we can recover some existing algorithms from~\cite{BJK+25} for $\SLWE$, using our algorithm for $\ISIS$ and our reverse reduction. 
\end{enumerate}

Our results clarify the relation between $\SLWE$ and $\ISIS$, showing both their strong connection and the remaining barriers to proving full equivalence. We now provide a more formal presentation of our contributions.

\subsection{Problems studied}
We first informally define some of the problems we consider in this work in order to properly state our contributions. We provide a formal and more in-depth description of these problems and their relation to other lattice-based problems in \Cref{Section:Definitions}.
The problems in this subsection are all parameterized by a matrix $\Am \in \ZZ_q^{n \times m}$, and either by an amplitude function $f : \Z_q^m \rightarrow \mathbb{C}$ such that $\norm{f}_2 = 1$, or by a set $T \subset \ZZ^m$, where $q,n,m$ are positive integers.

\begin{definition}[$\ISIS(\Am, T)$]
Given a uniformly random vector $\yv \Unif \ZZ_q^n$ as input, the $\ISIS(\Am, T)$ problem asks to find a vector $\xv \in T$ such that $\Am \xv = \yv \mod q$.
\end{definition}
This problem is usually defined for matrices $\Am$ chosen at random and the set $T$ chosen as $\{\xv \in \Z_q^m : \norm{\xv}_2 \le \beta\}$ for a parameter $\beta$. Our results will hold for any choice of matrices and of set $T$, which makes them more general. Chen, Liu, and Zhandry introduced variants of the canonical Learning With Errors problem, namely Search-LWE (\SLWE) and Construct-LWE (\CLWE).

\begin{definition}[$\SLWE(\Am,f)$]
Given the state $\ket{\psi_\sv} = \sum_{\ev \in \Z_q^m} f(\ev) \ket{\Am\T\sv + \ev}$ as input, where $\sv \Unif \ZZ_q^b$ is a uniformly random secret vector, the $\SLWE(\Am,f)$ problem asks to recover $\sv$.
\end{definition}

\begin{definition}[$\CLWE(\Am,f)$]
The $\CLWE(\Am,f)$ problem asks to construct the state
$$
	\ket{W} = \frac{1}{\sqrt{Z}} \sum_{{\sv \in \Z_q^n }}\sum_{\ev \in \Z_q^m} f(\ev) \ket{\Am\T\sv + \ev}
$$
where $Z$ is a normalization factor.
\end{definition}
In this work, we also introduce an inhomogeneous variant of the Construct LWE problem. 
\begin{definition}[$\ICLWE(\Am,f)$]
Given a uniformly random vector $\yv \Unif \ZZ_q^n$ as input, the $\ICLWE(\Am,f)$ problem asks to construct the unitary $\ket{\yv}\ket{\zerov} \rightarrow \ket{\yv}\ket{W_\yv}$, where we define
$$
	\ket{W_\yv} = \frac{1}{\sqrt{w_\yv}} \sum_{{\sv \in \Z_q^n }}\sum_{\ev \in \Z_q^m}  \omega^{- \yv \cdot \sv}f(\ev) \ket{\Am\T\sv + \ev}
$$
with normalizing factors $w_\yv$.
\end{definition}

\subsection{Contributions}
The goal of this paper is to investigate the two-way  reduction between $\ISIS$ and $\SLWE$. A graphical summary of our results is presented in \Cref{Figure:1}.
	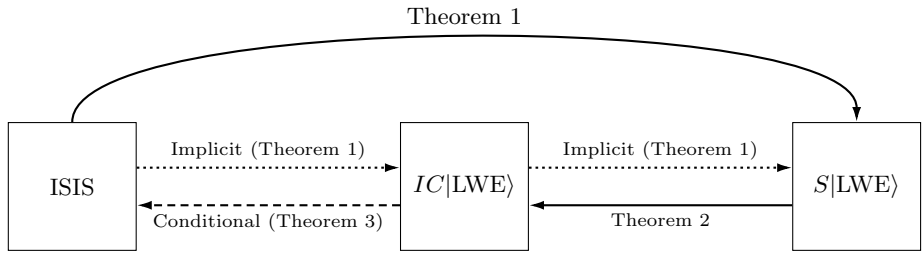
\begin{figure}
	\begin{center}
		\begin{tikzpicture}[
			square/.style={draw, minimum size=1.7cm, align=center},
			solidarrow/.style={-{Latex[length=2mm,width=1.2mm]}, thick},
			dottedarrow/.style={-{Latex[length=2mm,width=1.2mm]}, thick, dotted, dash pattern=on 1pt off 1.5pt},
			dashedarrow/.style={-{Latex[length=2mm,width=1.2mm]}, thick, dashed, dash pattern=on 3pt off 2pt},
			globallabelabove/.style={font=\small, inner sep=0.8pt, above=4pt}, %
			globallabelbelow/.style={font=\small, inner sep=0.8pt, below=4pt}, %
			innerlabelabove/.style={font=\scriptsize, inner sep=0.5pt, above=2pt},
			innerlabelbelow/.style={font=\scriptsize, inner sep=0.5pt, below=2pt},
			node distance=3.5cm
			]

			\node[square] (isis) {ISIS};
			\node[square, right=of isis] (clwe) {$\ICLWE$};
			\node[square, right=of clwe] (slwe) {$\SLWE$};

			\draw[solidarrow] (isis.north) .. controls +(0,1.5) and +(0,1.5) .. (slwe.north)
			node[midway, globallabelabove] {Theorem~1};

			\draw[dottedarrow] ([yshift=0.25cm]isis.east) -- ([yshift=0.25cm]clwe.west)
			node[midway, innerlabelabove] {Implicit (Theorem~1)};
			\draw[dottedarrow] ([yshift=0.25cm]clwe.east) -- ([yshift=0.25cm]slwe.west)
			node[midway, innerlabelabove] {Implicit (Theorem~1)};

			\draw[solidarrow] ([yshift=-0.25cm]slwe.west) -- ([yshift=-0.25cm]clwe.east)
			node[midway, innerlabelbelow] {Theorem~2};
			\draw[dashedarrow] ([yshift=-0.25cm]clwe.west) -- ([yshift=-0.25cm]isis.east)
			node[midway, innerlabelbelow] {Conditional (Theorem~3)};
			
		\end{tikzpicture}
	\end{center}
	\caption{Table of reductions between the different studied problems. An arrow $A \rightarrow B$ means that $A$ reduces to $B$, or in other terms that an efficient algorithm for problem $B$ can be used to construct an efficient algorithm for problem $A$.}
	\label{Figure:1}
\end{figure}
Our first contribution relates to the forward reduction $\ISIS \rightarrow \SLWE$.
\begin{theorem}[Informal]
	If we have an efficient algorithm solving $\SLWE(\Am,f)$ with $Supp(\hf) \subseteq T$, then we can construct an efficient algorithm for $\ISIS(\Am,T)$.
\end{theorem}

This theorem works also when the algorithm solves $\SLWE(\Am,f)$ with some non-negligible probability $p$ and we also show how to relax the constraint $Supp(\hf) \subseteq T$. The only requirement, which appears in all the works using this reduction, is that the state $\sum_{\ev \in \Z_q^m} f(\ev)\ket{\ev}$ is efficiently sampleable.

In order to prove this result, we start from the reduction of~\cite{CT25} and make a few notable changes. Our main goal was to ensure the reduction of~\cite{CT25} holds for any quantum algorithm solving $\SLWE(\Am,f)$, since several applications actually use a full quantum algorithm here. First, we change slightly the way the algorithm works in order to be used with quantum algorithms. Our main contribution then is to change the error analysis occurring from an error in the $\SLWE(\Am,f)$ algorithm. We fully take advantage of the fact that we want to solve the inhomogeneous variant $\ISIS(\Am,T)$ in order to perform a tighter analysis.  Another thing to note is that we rephrase this reduction in terms of lattice problems while~\cite{CT25} phrases the reduction in terms of coding problems. There is a straightforward correspondence between the two but our formulation makes it easier to relate with existing work on the subject. 

The above theorem is very general but a natural question is how to choose the function $f$ and the set $T$? We provide examples that are standard in the literature. Let $\rho_{r}$ be the discrete Gaussian distribution on $\Z_q$ of parameter $r$. The $\LWE$ problem is most often defined with the error distribution $\rho_r^m$, which means the corresponding natural quantum problem is $\SLWE(\Am,\sqrt{\rho_r^m})$. On the other hand, the most standard setting for $\ISIS$ is to consider $T = \{\xv \in \F_q^n : \norm{\xv}_2 \le l\}$, in which case the corresponding problem is denoted $\ISIS_2(\Am,l)$. Applying our theorem with the Gaussian error function and $T$ we can show
\begin{proposition}\label{Proposition:Gaussians}
	Let $\Am \in \F_q^{n \times m}$ and a real number $r > 0$. If we have an efficient quantum algorithm for $\SLWE(\Am,\sqrt{\rho_r^m})$, then we have an efficient quantum algorithm for $\ISIS_2(\Am,\sqrt{m} \frac{q}{\sqrt{8\pi}r})$. In particular, if we have an efficient algorithm for $\LWE(\Am,{\rho_r^m})$, then we have an efficient quantum algorithm for $\ISIS_2(\Am,\sqrt{m} \frac{q}{\sqrt{8\pi}r})$. 
\end{proposition}

\COMMENT{We can directly apply this theorem with standard lattice-based parameters. Informally, if we consider a discrete Gaussian on $\Z_q$ with standard deviation $\sigma$, which we denote $\chi_\sigma$, then an algorithm for solving $\SLWE(\Am,\chi_\sigma^m)$ for a randomly chosen matrix $\Am \in \Z_q^{n \times m}$ can be used to solve $\ISIS(\Am,T)$ with $T = \{\xv \in \Z_q^m : \norm{\xv}_2 \le \frac{q\sqrt{m}}{\sigma}\}$ at least for $\sigma \ll q$. This is directly implied by our theorem, since $\widehat{\chi_\sigma} \approx \chi_{\frac{q}{\sigma}}$ and the distribution $\chi_{\frac{q}{\sigma}}^m$ is highly concentrated around words of norm $\sqrt{m}\frac{q}{\sigma}$.}

 Then, we study the reverse reduction. Our reduction $\ISIS \rightarrow \SLWE$ goes through the intermediate problem which is $\ICLWE$. We do not formally write it as a reduction $\ISIS \rightarrow \ICLWE$ and $\ICLWE \rightarrow \SLWE$ because it would induce a loss in some of the parameters but our \Cref{Theorem:1} actually implicitly uses this route. Our first result for the reverse reduction is that $\SLWE$ reduces to $\ICLWE$.

\begin{theorem}[informal]
	If we have an efficient quantum algorithm that solves $\ICLWE(\Am,f)$ then we have an efficient quantum algorithm for $\SLWE(\Am,f)$.
\end{theorem}

This theorem does not require anything on the algorithm for $\ICLWE(\Am,f)$, and is robust to errors in the algorithm. One thing to take into account is that the problem $\SLWE(\Am,f)$ can be intractable from an information theoretic point of view for certain choices of $f$. We explicit this regime and show that our theorem holds for any $f$ such that the problem $\SLWE(\Am,f)$ is tractable.

Finally, we investigate the relation between $\ICLWE(\Am,f)$ and $\ISIS(\Am,T)$. Recall that in $\ICLWE(\Am,f)$, we want to compute the unitary $\ket{\yv}\ket{\zerov} \rightarrow \ket{\yv}\ket{W_\yv}$. We can compute the Fourier transform of $\ket{W_\yv}$.
$$ \ket{\widehat{W_\yv}} \sim \sum_{\xv \in \Lambda_\yv^\bot(\Am)} \hf(\xv) \ket{\xv},$$
where $\Lambda_\yv^\bot(\Am) = \{\xv \in \Z^m : \Am \xv = \yv \mod q\}$ denotes the $q$-ary shifted dual lattice associated to $\Am$.
On the other hand an algorithm for $\ISIS(\Am,T)$ outputs a string $\xv \in \Lambda_\yv^\bot(\Am) \cap T$, so if we take $\hf = \one_T$, the two problems look quite similar. Unfortunately, even from an algorithm that outputs  a uniformly random element of $\Lambda_\yv^\bot(\Am) \cap T$, it is not clear how to construct the state $\ket{\widehat{W_\yv}}$.

Under some condition on the algorithm used for $\ISIS$, we can prove the missing part of the reduction, namely that $\ICLWE \rightarrow \ISIS$. More precisely, we introduce the notion of full support randomness-recoverable algorithms, which are randomized classical algorithms from which we can recover the randomness from the solution and that has an output distribution which is uniform over the solutions.  We manage to prove the following
\begin{theorem}[Informal]\label{Theorem:3}
	Assume we have an efficient classical algorithm for $\ISIS(\Am,T)$ which is full support randomness-recoverable, then we have an efficient quantum algorithm for $\ICLWE(\Am,f)$, with $\hf = \one_T$. Using \Cref{Theorem:2}, this means we can also obtain an algorithm for $\SLWE(\Am,f)$.
\end{theorem}

While these conditions on the algorithm for $\ISIS$ seem very strict and potentially unachievable for real algorithms, we show that this is not the case. Our final contribution is to show how to recover an algorithm of~\cite{BJK+25}. This result was phrased in terms of the $\overline{\text{EDCP}}$ problem, which is a variant of the \emph{Extrapolated Dihedral Coset Problem}.

\begin{definition}[$\SEDCP(\Am,T)$]
	Let $q,n,m$ be positive integers, a set $T \subseteq \Z_q^m$ and a matrix $\Am \in \F_q^{n \times m}$. We sample $\sv \Unif \Z_q^n$. Given $\ket{\phi_\sv} = \frac{1}{\sqrt{|T|}}\sum_{\yv \in T} \omega_q^{\Am\T\sv \cdot \yv}\ket{\yv}$, the goal is to recover $\sv$
\end{definition}

In~\cite{BJK+25}, this problem is defined where $T = R^m$ with $R \subseteq \Z_q$. In this case, we can rewrite 
$\ket{\phi_{\sv}} = \bigotimes_{i = 1}^m \left(\frac{1}{\sqrt{|R|}}\sum_{j \in R} \omega_q^{j (\av_i \cdot \sv)}\ket{j}\right),$
where $\av_i$ is the $i^{th}$ line of $\Am\T$.

How is this problem related to our reverse reduction? First notice that $\SEDCP(\Am,T)$ is equivalent to the $\SLWE(\Am,f)$ problem with $\hf = \frac{1}{\sqrt{|T|}} \one_T$. Indeed, if we write the states arising in the $\SLWE$ problem $\ket{\psi_{\sv}}  = \sum_{\ev \in \Z_q^m} f(\ev) \ket{{\Am}\T\sv + \ev}$, we can compute 
\begin{align*}
	\widehat{\ket{\psi_{\sv}}} = \sum_{\yv \in \Z_q^m} \omega_q^{\Am\T\sv \cdot \yv}\left(\frac{1}{\sqrt{q^m}}\sum_{\ev \in \Z_q^m} \omega_q^{\yv \cdot \ev}f(\ev)\right)\ket{\yv} = \sum_{\yv \in \Z_q^m} \omega_q^{\Am\T\sv \cdot \yv} \hf(\yv)\ket{\yv} =\ket{\phi_\sv},
\end{align*}

where $\ket{\phi_\sv}$ is the state used in $\SEDCP$. This means we can go from one problem to the other just by applying a quantum Fourier transform on our input state. Our reverse reduction can therefore be naturally restated with this problem

\begin{theorem}\label{Theorem:3bis}
	Assume we have an efficient classical algorithm for $\ISIS(\Am,T)$ which is full support randomness recoverable, then we have an efficient quantum algorithm for $\SEDCP(\Am,T)$.
\end{theorem}

One of the main results of Bai et al. is to prove the following statement

\begin{proposition}[\cite{BJK+25}]\label{Proposition:BJK+25}
	For parameters $n,m,q = 2^l$ with $q = \poly(n)$ and $m = 2^{O(\log(n)\log(q))}$, there exists a quantum algorithm for
	$\SEDCP(\Am,\Z_2^m)$ running in time $2^{O(\log(n)\log(q))}$ in the case $\Am$ is randomly chosen in $\Z_q^{n \times m}$.
\end{proposition}

They also show that the $\emph{Extrapolated Dihedral Coset Problem}$ they consider naturally reduces to $\SEDCP$. The quantum algorithm used to prove this proposition is a Kuperberg style algorithm but strongly leverages on the fact that $q$ is a small power of $2$, with ideas that were already developed by Bonnetain and Naya-Plasencia~\cite{BN18} in a somewhat different context.

When looking at this algorithm carefully, this quantum algorithm also strongly resembles a classical algorithm for $\SIS_{\infty}$ in the same parameter regime (by replacing $n$ with $m-n$). This algorithm was presented in~\cite[Appendix A of the ArXiv version]{CLZ22}, and attributed to Regev. A natural question therefore becomes whether this algorithm can be used in our reverse reduction

We show that it is actually possible to modify this classical algorithm to the case of $\ISIS$ and such that it satisfies all the requirements of \Cref{Theorem:3}.  

\begin{theorem}[Informal]\label{Theorem:4}
	One can adapt the classical algorithm for $\ISIS$ by Regev to make it randomness-recoverable and full support. Plugging this algorithm into \Cref{Theorem:3bis}, we can recover \Cref{Proposition:BJK+25}.
\end{theorem}

\subsection{Comparison with Previous Works}
\label{sec:comp-previous-works}
As mentioned above, the authors of \cite{CT25} present a reduction from $\ISIS$ to $\SLWE$ for some class of $\SLWE$ solvers.
Our result generalizes theirs as it works for all $\SLWE$ solvers. Our formulation in terms of $Supp(\hf)$ makes is also very powerful, and this was crucially used in a subsequent work \cite{chailloux2025opi} (and adapted to codes) to improve Decoded Quantum Interferometry~\cite{JSW+24}.

Additionally, in terms of proof techniques, we provide a much sharper analysis of the errors %
Instead of just adding an error term that depends on the solver's success probability, we use the structure of the error vectors and integrate this in the final success probability.
Note that this is done by not going through the quantum sampling problem $\ICLWE$ but directly studying the reduction between $\ISIS$ and $\SLWE$.

\paragraph{Chain of reductions from $\SLWE$ to $\ISIS$.}
A natural question is to ask whether reduction from$\SLWE$ to $\ISIS$ existed in the literature. Such a reduction can actually be obtained  from the following chain of reductions: $\SLWE \rightarrow \LWE \rightarrow d\LWE \rightarrow \SIS \rightarrow \ISIS$. While the $\SLWE \rightarrow \LWE$ and $\SIS \rightarrow \ISIS$ reductions are lossless, the two other ones are lossy, and importantly, this chain is one-way. 

The main difference is that our reduction is lossless, in the sense that the matrix $\Am$ stays the same throughout the reductions. Moreover, in the regimes where our reverse reduction holds, we can apply the forward and the reverse reduction to recover our original problem, without any losses in the parameters. In this sense, our reverse reduction can be seen as the natural reduction that targets directly $\SLWE$. In particular, we couldn't have recovered the results of \cite{BJK+25} if we used the classical reduction instead of our reduction. 

\subsection{Takeaways and Future Work}

Our first contribution is important as it provides what we hope to be the final form of the reduction $\ISIS \rightarrow \SLWE$. A natural direction for future work is to extend this reduction to codes, which we believe will help address open questions related to the \emph{Decoded Quantum Interferometry} framework. As we mentioned, this result was used in~\cite{chailloux2025opi} to improve on quantum algorithms for the Optimal Polynomial Intersection Problem ($OPI$), as originally defined by~\cite{JSW+24}. 

From a conceptual standpoint, the reverse reduction is new. This is the first time there is an (even conditional) tight reduction between these two problems. It shows that one cannot fundamentally improve this framework based on Regev's reduction if we restrict ourselves to $\SLWE$ and $\ISIS$. When the reduction was phrased between $\ISIS$ and $\LWE$, such a reverse reduction could not be envisioned. Moreover, we showed that the intermediate problem $\ICLWE$ is key to understanding the relation between these problems. 

Our results could also be used to construct other algorithms for $\SLWE$ (or equivalently $\overline{\text{EDCP}}$ defined in~\cite{BJK+25}). We have already recovered some of these results showing that the randomness-recoverability conditions are achievable. It would be interesting to explore other algorithms, such as the one by Chen et al.~\cite{CHL+25}, and investigate whether they fit into our framework. Finally, finding new algorithms for $\SLWE$ could have significant consequences, since there is a reduction $\LWE \rightarrow \SLWE$ shown in~\cite{BKSW18}. While our reverse reduction has some limitations, it nevertheless paves the way for many potential new algorithms for $\SLWE$, obtained via known algorithms for $\ISIS$-type problems and incorporated into the reduction of~\cite{BKSW18}.

\COMMENT{ We also show present a modified construction that allows us to deal with any $q$ (but for which the condition on the infinity norm will be better the more $q$ is composite). The idea also comes from the description of the algorithm~\cite{CLZ22} but we had to tweak it so that it can be extended to the case of $\ISIS$.

\begin{theorem}\label{Theorem:5}
		For some parameters $n,m,q$ with $q = \prod_{i = 1}^k p_i = \poly(n)$ and $m = 2^{O(\log(n)\log(q))}$, there exists an algorithm for $\SLWE(\Am,f)$ in the case $\Am$ is randomly chosen in $\Z_q^{n \times m}$ and $f$ is the function such that $\hf = \one_{T}$ where 
		$$ T = \left\{\xv \in \Z_q^m : \norm{\xv}_{\infty} \le \sum_{i = 1}^k {\prod_{j = 1}^i \lfloor \frac{p_j}{2} \rfloor}\right\}.$$
\end{theorem}
\paul{i changed the bound to $\sum_{i = 1}^k {(k - i + 1) \lfloor \frac{p_j}{2} \rfloor}$, harder to read but tighter. not sure it's better}

Because of the adaptation to $\ISIS$, we actually have a little loss in the infinity norm constraint. This result to some extent solves a technical issue presented in~\cite{BJK+25}, where their techniques were limited to the case where $q$ is a power of two. This theorem however still cannot be used directly to construct efficient quantum algorithms for $\LWE$ via the~\cite{BKSW18} reduction. }

\subsection*{Acknowledgements}
\if \isanonymous 0{%
	We acknowledge funding from the French PEPR integrated projects EPIQ (ANR-22-PETQ-007), PQTLS (ANR-22-PETQ-008) and HQI (ANR-22-PNCQ-0002) all part of plan France 2030.}
\else {}
\fi
We thank the anonymous reviewers for their insightful comments on our previous submission.
\section{Preliminaries}
\label{sec:preli}
\setcounter{theorem}{0}
\subsection{Notations}
For a positive integer $q$, we write $\Z_q$ for $\Z/{q\Z}$. We write $\Z_q = \{-\lfloor\frac{q}{2}\rfloor,\dots,\lfloor\frac{q-1}{2}\rfloor\}$. Vectors with elements in $\Z_q$ will be denoted with bold small letters such as $\xv,\yv$ and matrices with elements in $\Z_q$ will be denoted with capital bold letters such as $\Am,\Hm$.
For any function $f: \ZZ_q^m \rightarrow \CC$, we define the $\ell_2$ norm $\Vert f \Vert_2$ of $f$ as $\Vert f \Vert_2 = \sum_{\xv \in \ZZ_q^m} |f(\xv)|^2$.

The canonical $q$-th root of unity is denoted $\omega_q = e^{2i\pi/q}$. We will consider only roots of unity $\omega_q$ for the alphabet size $q$ and will usually omit the subscript $q$.

{For any probability distribution $D$, we write $x \gets D$ to indicate that $x$ is sampled according to $D$.
We abuse the notation and write $x \gets S$ for any set $S$ to indicate that $x$ is sampled uniformly from $S$.}
We use the following definition for statistical distance between any two probability distributions represented by their probability mass functions $p$ and $q$: $\Delta(p, q) = \frac{1}{2} \sum_x |p(x) - q(x)|$.

\begin{definition}
	Let $q,n,m$ be positive integers and a matrix $\Am \in \Z_q^{n \times m}$. For each $\yv \in \Z_q^n$, we define the shifted dual lattice
	$$ \Lambda^\bot_\yv(\Am) = \{\xv \in \Z^m : \Am\T \xv = \yv \mod q \}.$$
\end{definition}

\paragraph{Fourier Transform and Quantum Fourier Transform}
For a function $f : \Z_q \rightarrow \mathbb{C}$, we define its Fourier transform 
$$ \hf(x) = \frac{1}{\sqrt{q}} \sum_{y \in \Z_q} \omega^{xy} f(y).$$
We extend this definition to functions $f : \Z_q^m \rightarrow \mathbb{C}$ and denote the Fourier transform of $f$ as $\hf(\xv) = \frac{1}{\sqrt{q^m}} \sum_{\yv \in \Z_q^n} \omega^{\xv \cdot \yv} f(\yv).$
The Quantum Fourier Transform on $\Z_q$ is the unitary operations
$$ QFT_{\Z_q}(\ket{x}) = \ket{\widehat{x}} = \frac{1}{\sqrt{q}}\sum_{y \in \Z_q} \omega^{xy} \ket{y}.$$
Again, it is extended to a unitary acting on $\Z_q^m$ as follows: $QFT_{\Z_q^m}(\ket{\xv}) = \ket{\widehat{\xv}} = \frac{1}{\sqrt{q^m}} \sum_{\yv \in \Z_q^m} \omega^{\xv \cdot \yv}\ket{\yv}$.
{We omit the subscript and simply write $QFT$ when clear from the context.}

\subsection{Definition of computational problems}\label{Section:Definitions}
The Short Integer Solution problem is one of the cornerstones of lattice-based cryptography. 
\begin{definition}[Short Integer Solution $\SIS(\Am,\beta)$]
	Let $q,n,m$ be positive integers, and $\beta$ and a matrix $\Am \in \Z_q^{n \times m}$. The goal is, given $\Am$, to find $\xv \in \Z^m \backslash \{\zerov\}$ such that $\Am \xv = \zerov \mod q$ and $\norm{\xv}_2 \le \beta$.
\end{definition}

The above problem is usually defined for a uniformly random $\Am \in \Z_q^{n \times m}$ but is also defined for structured matrices and/or with other norms such as the infinity norm. This problem also has what is called an inhomogeneous variant

\begin{definition}[Inhomogeneous Short Integer Solution $\ISIS(\Am,\beta)$]
	Let $q,n,m,\beta$ be positive integers and a matrix $\Am \in \Z_q^{n \times m}$. We sample $\yv \Unif \Z_q^n$ and the goal is, given $(\Am,\yv)$, to find $\xv \in \Z^m$ such that $\Am \xv = \yv \mod q$ and $\norm{\xv}_2 \le \beta$.
\end{definition}

We generalize these definitions by replacing the condition $\norm{\xv}_2 \le \beta$ with the condition $\xv \in T$ for a given set $T$. 
\begin{definition}[Short Integer Solution $\SIS(\Am,T)$]
	Let $q,n,m$ be positive integers,  a matrix $\Am \in \Z_q^{n \times m}$ and $T \subseteq \Z^m$. The goal is, given $\Am$ and $T$, to find $\xv \in T\backslash \{\zerov\}$ such that $\Am \xv = \zerov \mod q$.
\end{definition}

\begin{definition}[Inhomogeneous Short Integer Solution $\ISIS(\Am,T)$]
	Let $q,n,m$ be positive integers,  a matrix $\Am \in \Z_q^{n \times m}$ and $T \subseteq \Z^m$. We sample $\yv \Unif \Z_q^n$ and the goal is, given $(\Am,\yv)$ and $T$, to find $\xv \in  T$ such that $\Am \xv = \yv \mod q$.
\end{definition}

We now define a general form of the Learning with errors problem,  for an alphabet size $q$, dimension $n$ and number of samples $m$.
\begin{definition}[Learning With Errors $\LWE(\Am,p)$]
	Let $q,n,m$ be positive integers, a probability distribution $p$ on $\Z_q^m$ and a matrix $\Am \in \Z_q^{n \times m}$. We sample $\sv \Unif \Z_q^n$ and $\ev \Unif p$. Given $(\Am,\Am\T\sv + \ev)$, the goal is to recover $\sv$.
\end{definition}
This problem is commonly studied in the case where $\Am$ is uniformly chosen in $\Z_q^{n \times m}$ and $p = \rho_r^m$ where $\rho_r$ is a discretized Gaussian distribution on $\Z_q$ of parameter $r$. We now define the variants explicitly introduced in~\cite{CLZ22}, where the noise is in quantum superposition.

\begin{definition}[$\SLWE(\Am,f)$]
	Let $q,n,m$ be positive integers, a function $f : \Z_q^m \rightarrow \mathbb{C}$ such that $\norm{f}_2 = 1$ and a matrix $\Am \in \Z_q^{n \times m}$. We sample $\sv \Unif \Z_q^n$. Given $\ket{\psi_\sv} = \sum_{\ev \in \Z_q^m} f(\ev) \ket{\Am\T\sv + \ev}$, the goal is to recover $\sv$.
\end{definition}

Notice that by measuring $\ket{\psi_\sv}$, one can recover a random $\Am\T\sv + \ev$ for $\ev \Unif |f|^2$. This immediately implies that $\SLWE(\Am,f)$ is easier than $\LWE(\Am,|f|^2)$. 

\begin{definition}[$\CLWE(\Am,f)$]
	Let $q,n,m$ be positive integers, a function $f : \Z_q^m \rightarrow \mathbb{C}$ such that $\norm{f}_2 = 1$ and a matrix $\Am \in \Z_q^{n \times m}$. The goal is to construct the unit vector 
	$$ \ket{W} = \frac{1}{\sqrt{Z}} \sum_{{\sv \in \Z_q^n }}\sum_{\ev \in \Z_q^m} f(\ev) \ket{\Am\T\sv + \ev},$$
	where $Z$ is a normalization factor.
\end{definition}

In this work, we introduce an inhomogeneous variant of this problem, which will be useful for our reductions. 

\begin{definition}[$\ICLWE(\Am,f)$]
	Let $q,n,m$ be positive integers, a function $f : \Z_q^m \rightarrow \mathbb{C}$ such that $\norm{f}_2 = 1$ and a matrix $\Am \in \Z_q^{n \times m}$. The goal is to construct the unitary 
	$$\ket{\yv}\ket{\zerov} \rightarrow \ket{\yv}\ket{W_\yv},$$
	where we define 
	$$ \ket{W_\yv} = \frac{1}{\sqrt{w_\yv}} \sum_{{\sv \in \Z_q^n }}\sum_{\ev \in \Z_q^m}  \omega^{- \yv \cdot \sv}f(\ev) \ket{\Am\T\sv + \ev},$$
	with normalizing factors $w_\yv$.
\end{definition}

\begin{remark} \label{Remark:1}
  Defining $\ICLWE$ through the construction of \emph{a unitary} instead of an algorithm constructing $\ket{W_\yv}$ on input $\yv$ might seem an unexpected generalization of the $\CLWE$ problem --- as the existence of such an algorithm does not imply the unitary we define, but merely one that satisfies $\ket{\yv}\ket{\zerov}\ket{\zerov} \rightarrow \ket{\yv}\ket{W_\yv}\ket{\phi_\yv}$ for some garbage state $\ket{\phi_\yv}$.
  However, we note that in the special case of $\CLWE$, defining the problem through the construction of a unitary or an algorithm is equivalent (in this sense, we do generalize this problem), and that we use $\ICLWE$ as an intermediate step in our reduction, and successfully manage to instantiate our final reduction with this condition.
	Furthermore, defining this problem through a unitary turns out to be crucial in our reduction, as it allows to apply the reverse unitary.
\end{remark}

\section{Preliminary calculations around $\SLWE$ and $\ICLWE$ and tractability results}
The goal of this section is to present some preliminary calculations on the $\SLWE,\ICLWE$ and $\ISIS$ problems. We also provide a discussion on the tractability regime of $\SLWE$ from an information theoretic view, which is a direct generalization of results in~\cite{CT24}. 

\subsection{Preliminary calculations}

An important calculation will be to write the states $\ket{\psi_\sv}$ and $\ket{W_\yv}$ appearing respectively in $\SLWE$ and $\ICLWE$ in the Fourier basis.  

\begin{proposition}\label{Proposition:PsiFourier}
	Let $q,n,m$ be positive integers, a function $f : \Z_q^m \rightarrow \mathbb{C}$ such that $\norm{f}_2 = 1$ and a matrix $\Am \in \Z_q^{n \times m}$. For each $\sv \in \Z_q^n$, we define $\ket{\psi_\sv} = \sum_{\ev \in \Z_q^m} f(\ev)\ket{\Am\T \sv + \ev}$. We have 
	$$\ket{\widehat{\psi_{\sv}}} = \sum_{\yv \in \Z_q^n} \omega^{\yv \cdot \sv}  \sum_{\xv \in \Lambda^{\bot}_{\yv}(\Am)}  \hf(\xv) \ket{\xv}.$$
\end{proposition}
\begin{proof}
	We write 
	\begin{align*}
		\ket{\widehat{\psi_{\sv}}} & = \frac{1}{\sqrt{q^m}}\sum_{\xv \in \Z_q^m}\sum_{\ev \in \Z_q^m} \omega^{\xv \cdot (\Am\T\sv + \ev)} f(\ev)\ket{\xv} \\
		& = \frac{1}{\sqrt{q^m}}\sum_{\xv \in \Z_q^m} \omega^{\xv \cdot \Am\T\sv} \sum_{\ev \in \Z_q^m} \omega^{\xv \cdot \ev} f(\ev)\ket{\xv} \\
		& = \frac{1}{\sqrt{q^m}}\sum_{\xv \in \Z_q^m} \omega^{\Am \xv \cdot \sv} \sum_{\ev \in \Z_q^m} \omega^{\xv \cdot \ev} f(\ev)\ket{\xv} \\	
		& = \sum_{\xv \in \Z_q^m} \omega^{\Am \xv \cdot \sv} \hf(\xv)\ket{\xv} \\		
		& = \sum_{\yv \in \Z_q^n} \omega^{\yv \cdot \sv}  \sum_{\xv \in \Lambda^{\bot}_{\yv}(\Am)}  \hf(\xv) \ket{\xv}
	\end{align*}
\end{proof}

\begin{proposition}\label{Proposition:WFourier}
	Let $q,n,m$ be positive integers, a function $f : \Z_q^m \rightarrow \mathbb{C}$ such that $\norm{f}_2 = 1$ and a matrix $\Am \in \Z_q^{n \times m}$. For each $\yv \in \Z_q^n$, let
	$$ \ket{W_\yv} = \frac{1}{\sqrt{w_\yv}} \sum_{{\sv \in \Z_q^n }}\sum_{\ev \in \Z_q^m}  \omega^{- \yv \cdot \sv}f(\ev) \ket{\Am\T\sv + \ev}.$$
	Then $\widehat{\ket{W_\yv}} = \frac{q^n}{\sqrt{w_\yv}} \sum_{\xv \in \Lambda^\bot_\yv(\Am)} \hf(\xv)\ket{\xv}.$
\end{proposition}
\begin{proof}
	We write
		\begin{align*}
			\ket{\widehat{W_\yv}} & = \frac{1}{\sqrt{q^m}}\frac{1}{\sqrt{w_\yv}} \sum_{\xv \in \Z_q^m} \sum_{\sv \in \Z_q^n} \sum_{\ev \in \Z_q^m}\omega^{\xv \cdot (\Am\T \sv + \ev)} \omega^{-\yv \cdot \sv} f(\ev)\ket{\xv} \\
			& = \frac{1}{\sqrt{q^m}}\frac{1}{\sqrt{w_\yv}} \sum_{\xv \in \Z_q^m}\sum_{\sv \in \Z_q^n} \omega^{\xv \cdot \Am\T\sv - \yv \cdot \sv}\sum_{\ev \in \Z_q^m} \omega^{\xv \cdot \ev}f(\ev)\ket{\xv} \\
			& = \frac{1}{\sqrt{w_\yv}} \sum_{\xv \in \Z_q^m} \sum_{\sv \in \Z_q^n} \omega^{(\Am \xv - \yv)\cdot \sv} \hf(\xv)\ket{\xv} \\
			& = \frac{q^n}{\sqrt{w_\yv}} \sum_{\xv \in \Lambda^\bot_\yv(\Am)} \hf(\xv)\ket{\xv}
	\end{align*}
  {where in the last equality follows from the identity $\sum_{\sv \in \Z_q^n} \omega^{(\Am \xv - \yv)\cdot \sv} = 0$ if $\Am \xv - \yv \neq 0$, and $q^n$ otherwise.}

\end{proof}
This proposition shows in particular that if one can construct the state $\ket{W_\yv}$ for any $\yv \in \Z_q^n$ with $Im(\hf) \subseteq T$ then one can solve the $\ISIS(\Am,T)$ by applying a Quantum Fourier Transform on $\ket{W_\yv}$ and measuring the resulting state in the computational basis. Also, this proposition directly implies that the $\ket{W_\yv}$ are pairwise orthogonal. 

Finally, we present relations between these states.
\begin{proposition}\label{Proposition:PsiToW}
	\begin{align*}
		\forall \yv \in \Z_q^n, \ \ket{W_\yv} & = \frac{1}{\sqrt{w_\yv}} \sum_{\sv \in \Z_q^n} \omega^{-\yv \cdot \sv} \ket{\psi_\sv} \\
		\forall \sv \in \Z_q^n, \ \ \ket{\psi_\sv} & = \frac{1}{q^n} \sum_{\yv \in \Z_q^n} \omega^{\yv \cdot \sv}\sqrt{w_\yv} \ket{W_\yv}
	\end{align*}
\end{proposition}
\begin{proof}
	The first equality comes directly from the definitions of $\ket{W_\yv}$ and $\ket{\psi_\sv}$. For the second equality, the above two propositions immediately imply that 
	$$ \ket{\widehat{\psi_\sv}} = \sum_{\yv \in \Z_q^n} \omega^{\yv \cdot \sv} \frac{\sqrt{w_\yv}}{q^n}\ket{\widehat{W_\yv}},$$
	which gives the result by performing an inverse Quantum Fourier Transform on each side of the equality. 
\end{proof}

Finally, we give a proposition related to the norms $w_\yv$.
\begin{proposition}
	$\E_{\yv}\left[w_\yv\right] = q^n$.
\end{proposition}
\begin{proof}
	Because $\ket{\widehat{W_\yv}}$ is a unit vector, we have 
	$$ w_\yv = q^{2n} \sum_{\xv \in \Lambda^\bot_\yv(\Am)} |\hf(\xv)|^2.$$
	Since $\norm{f}_2 = \norm{\hf}_2 = 1$, we immediately have 
	$$ \frac{1}{q^n}\sum_{\yv \in \Z_q^n} w_\yv = q^{n} \sum_{\yv \in \Z_q^n} \sum_{\xv \in \Lambda^\bot_\yv(\Am)} |\hf(\xv)|^2 = q^{n} \sum_{\xv \in \Z_q^m} |\hf(\xv)|^2 = q^{n}.$$ 
\end{proof}

\subsection{Tractability bound for $\SLWE(\Am,f)$}
\begin{proposition}
	Let $q,m,n$ be positive integers. Let $f : \Z_q^m \rightarrow \mathbb{C}$ such that $\norm{f}_2 = 1$ and let $\Am \in \Z_q^n$.  For each $\yv \in \Z_q^n$, let
	$$ \ket{W_\yv} = \frac{1}{\sqrt{w_\yv}} \sum_{{\sv \in \Z_q^n }}\sum_{\ev \in \Z_q^m}  \omega^{- \yv \cdot \sv}f(\ev) \ket{\Am\T\sv + \ev},$$
	where $w_\yv$ is a normalizing factor so that $\ket{W_\yv}$ are unit vectors. The maximum probability $p_{max}$ that a (potentially unbounded) quantum algorithm has of solving $\SLWE(\Am,f)$ is 
	$$ p_{max} = \left(\E_{\yv \Unif \Z_q^n} \left[\sqrt{\frac{w_\yv}{q^n}}\right]\right)^2.$$
\end{proposition}
\begin{proof}
	This is proven by analyzing the Pretty Good Measurement on the states $\ket{\psi_\sv}$, which turns out to be optimal for this family of states and then reusing the analysis of~\cite{CT24}. For completeness, we formally prove this proposition in \Cref{Appendix:Tractability}.
\end{proof}
This proposition explains why the term $\E_{\yv \Unif \Z_q^n} \left[\sqrt{\frac{w_\yv}{q^n}}\right]$ appears in some of our reductions. 
\section{The forward direction $\ISIS \rightarrow \SLWE$}
Our first result is to prove the general forward reduction $\ISIS \rightarrow \SLWE$.
\begin{theorem}\label{Theorem:1}
	Let $q,m,n$ be positive integers, let $\Am \in \Z_q^{n \times m}$ and let $f : \Z_q^n \rightarrow \mathbb{C}$ with $\norm{f}_2 = 1$. Let $T \subseteq \Z_q^m$. Assume that 
	\begin{itemize}
		\item There exists a quantum algorithm that solves $\SLWE(\Am,f)$ in time {$\Time_{\SLWE}$} and succeeds with probability $p$.
		\item $\sum_{\xv \in T} |\hf(\xv)|^2 = 1 - \eta$.
		\item $f$ is quantum samplable in time $\Time_{Sampl}$ {\ie} there is a quantum algorithm running in time $\Time_{Sampl}$ that constructs the state 
		$ \sum_{\ev \in \Z_q^m} f(\ev)\ket{\ev}.$
	\end{itemize}
	Then there exists a quantum algorithm that solves $\ISIS(\Am,T)$, that succeeds with probability at least $p(1-\eta) - 2\sqrt{p(1-p)\eta}$ and runs in time 
	$$\Time_{\ISIS} = O\left(\frac{1}{p}\left(\Time_{\SLWE} + \Time_{Sampl}\right) + \poly(m,\log(q))\right).$$
\end{theorem}
The remainder of this section is devoted to the proof of this theorem.

\subsection{Characterization of quantum algorithms for $\SLWE$}

A quantum algorithm for $\SLWE(\Am,f)$ can be described by a unitary $U$ (that depends on $\Am$ and $f$) such that for all $\sv \in \Z_q^n$,
$$ \ U \ket{\psi_\sv}\ket{0} = \sum_{\sv' \in \Z_q^n} \gamma_{\sv,\sv'} \ket{\sv'}\ket{\wpsi_{\sv,\sv'}}, \text{ for some unit vectors } \ket{\wpsi_{\sv,\sv'}} \text{ and } \gamma_{\sv,\sv'} \in \mathbb{C},$$
and the result is obtained by measuring the first register. 
The success probability of this algorithm for each $\sv$ is $p_\sv =|\gamma_{\sv,\sv}|^2$ and the overall success probability is $p = \frac{1}{q^n} \sum_{\sv} |\gamma_{\sv,\sv}|^2$. We first prove that any such quantum algorithm can be symmetrized in the sense that each $\gamma_{\sv,\sv}$ is equal to $\sqrt{p}$.
\begin{proposition}
	Let $\aa$ be an efficient quantum algorithm for $\SLWE(\Am,f)$ that succeeds with probability $p$. There exists an efficiently computable unitary $U$ such that for all $\sv \in \Z_q^n$,
	$$ U \ket{\psi_\sv}\ket{0} = \sum_{\sv' \in \Z_q^n} \gamma'_{\sv,\sv'} \ket{\sv'}\ket{\wpsi''_{\sv,\sv'}}, \text{ for some unit vectors } \ket{\wpsi''_{\sv,\sv'}} \text{and each } \gamma'_{\sv,\sv} = \sqrt{p}.$$
\end{proposition}
\begin{proof}
The idea is to use the symmetries inherent to the states $\ket{\psi_\sv}$. We present a first algorithm that succeeds with probability $p$ for each input state $\ket{\psi_\sv}$. Consider the shift unitaries $S_\zv : \ket{\xv} \rightarrow \ket{\xv + \zv}$ for $\xv,\zv \in \Z_q^m$ which are efficiently computable. Notice that $\forall \sv,\tv \in \Z_q^n$, we have $\ket{\psi_\tv} = S_{\Am\T (\tv - \sv)} \ket{\psi_\sv}$. We consider the following algorithm 
\begin{enumerate}
	\item Given input $\ket{\psi_\sv}$, construct  
	$$ \ket{\Omega_1} = \frac{1}{\sqrt{q^n}}\sum_{\tv \in \Z_q^n}S_{\Am\T \tv}\ket{\psi_\sv}\ket{0}\ket{\tv} = \sum_{\tv \in \Z_q^n}\ket{\psi_{\sv + \tv}}\ket{0}\ket{\tv}.$$
	\item  Apply $U$ on the first two register to obtain 
	$$ \ket{\Omega_2} = \frac{1}{\sqrt{q^n}}\sum_{\tv \in \Z_q^n} \sum_{\sv' \in \Z_q^n} \gamma_{\sv+\tv,\sv'} \ket{\sv'} \ket{\wpsi_{\sv + \tv,\sv'}}\ket{\tv}.$$
	\item We subtract the value from the third register in the first register to obtain
	\begin{align*}\ket{\Omega_3} & = \frac{1}{\sqrt{q^n}}\sum_{\tv \in \Z_q^n} \sum_{\sv' \in \Z_q^n} \gamma_{\sv+\tv,\sv'} \ket{\sv' - \tv} \ket{\wpsi_{\sv + \tv,\sv'}}\ket{\tv} \\
		& = \frac{1}{\sqrt{q^n}}\sum_{\sv' \in \Z_q^n}\sum_{\tv \in \Z_q^n} \gamma_{\sv + \tv,\sv'+\tv}\ket{\sv'} \ket{\wpsi_{\sv + \tv,\sv' + \tv}}\ket{\tv} \\
	\end{align*}
\end{enumerate}
If we measure the first register, we obtain $\sv$ with probability $\frac{1}{q^n} \sum_{\tv} |\gamma_{\sv + \tv,\sv + \tv}|^2 = p$ which is independent of $\sv$. If we perform the above algorithm fully coherently, we obtain a quantum unitary $U'$ such that for all $\sv \in \Z_q^n$,
$$U' \ket{\psi_\sv}\ket{0} = \sum_{\sv' \in \Z_q^n} \gamma'_{\sv,\sv'} \ket{\sv'}\ket{\wpsi'_{\sv,\sv'}}$$%
for some unit vectors $\ket{\wpsi_{\sv,\sv'}}$ and each $|\gamma'_{\sv,\sv}| = \sqrt{p}$.
In order to conclude, we just have to put the potential phases of $\gamma'_{\sv,\sv}$ into the second register so if we define $\ket{\wpsi''_{\sv,\sv'}} = \frac{\gamma_{\sv,\sv'}}{|\gamma_{\sv,\sv'}|}\ket{\wpsi'_{\sv,\sv'}}$, we can indeed write 
$$  U' \ket{\psi_\sv}\ket{0} = \sum_{\sv' \in \Z_q^n} \gamma'_{\sv,\sv'} \ket{\sv'}\ket{\wpsi''_{\sv,\sv'}} \quad \text{with each } \gamma'_{\sv,\sv'} = \sqrt{p}.$$
\end{proof}
\subsection{The main algorithm}\label{Section:Algorithm}
\subsubsection{Presentation of the main algorithm}

We present now a detailed description of our algorithm. We slightly modify the way it is presented in the literature in order to make the proofs easier.

	\begin{algobox}{Quantum algorithm based on Regev's reduction for \ICC}{reduction}
	\textbf{Input:} We start from a matrix $\Am \in \Z_q^{n \times m}$. Let $T \subseteq \Z^m$ and $f : \Z_q^m \rightarrow \mathbb{C}$ such that $\norm{f}_2 = 1$. For each $\sv \in \Z_q^n$, we write $\ket{\psi_{\sv}} = \sum_{\ev \in \Z_q^m} f(\ev)\ket{\Am\T \sv + \ev}$.  Assume we have an efficiently computable quantum unitary 
	$$ U \ket{\psi_\sv}\ket{0} = \sum_{\sv' \in \Z_q^n} \gamma_{\sv,\sv'} \ket{\wpsi_{\sv,\sv'}}\ket{\sv'}$$
	where $\forall \sv \in \Z_q^n, \ \gamma_{\sv,\sv} = \sqrt{\Pdec}$ and each $\norm{\ket{\wpsi_{\sv,\sv'}}} = 1$.
	Finally, we are given a random $\yv \in \Z_q^n$. \\
	\textbf{Goal:} Find $\xv \in \Lambda^{\bot}_{\yv}(\Am) \cap T$, where $\C^\bot_{\yv}(\Am) = \{\xv \in \Z^m :  {\Am}\xv = \yv \mod q\}$ \\ \\
	\textbf{Execution of the algorithm:}
	\begin{enumerate}
		\item First construct the state
		$\frac{1}{\sqrt{q^n}} \sum_{\sv \in \Z_q^n} \omega^{-\yv \cdot \sv} \ket{\psi_\sv}\ket{0}\ket{\sv}$ (See \Cref{Section:FirstAnalysis}).
		\item Perform the operation
		\begin{align*}
			\frac{1}{\sqrt{q^n}} \sum_{\sv \in \Z_q^n} \omega^{-\yv \cdot \sv} \ket{\psi_{\sv}}\ket{0}\ket{\sv} & \overset{\textcircled{\scalebox{0.7}{A}}}{\mathlarger{\mathlarger{\rightarrow}}} \frac{1}{\sqrt{q^n}}\sum_{\sv,\sv' \in \Z_q^n} \omega^{-\yv \cdot \sv} \gamma_{\sv,\sv'} \ket{\wpsi_{\sv,\sv'}}\ket{\sv'}\ket{\sv - \sv'}
		\end{align*}
		Here, $\textcircled{\scalebox{0.7}{A}}$ is done by applying $U$ on the first two registers and then subtracting the value of the second register in the third register. 
		\item Measure the third register. If we do not obtain $\zerov$, start again from step $1$. Otherwise, we obtain the state $\frac{1}{\sqrt{q^n}}\sum_{\sv \in \Z_q^n} \omega^{-\yv \cdot \sv} \ket{\wpsi_{\sv,\sv}}\ket{\sv}\ket{\zerov}$.
		\item Discard the third register and apply $U^\dagger$ on the first two registers. The resulting state is
		$$\ket{\Phi_{\yv}} = \frac{1}{\sqrt{q^n}}\sum_{\sv \in \Z_q^n} \omega^{-\yv \cdot \sv}\sqrt{\Pdec}\ket{\psi_\sv}\ket{\zerov} + \omega^{-\yv \cdot \sv} \sqrt{1 - \Pdec} \ket{Z_\sv}$$
		for  unit vectors $\ket{Z_\sv} \bot \ket{\psi_\sv}\ket{\zerov}$.
		\item We apply a Quantum Fourier Transform on the first register and measure in the computational basis. Output the outcome of the measurement.
	\end{enumerate}
\end{algobox}

\subsubsection{First analysis and running time of the algorithm}\label{Section:FirstAnalysis}
We first provide some details over each step of the algorithm and on the running time. Let 
\begin{itemize}
	\item $\Time_{\SLWE}$ be the running time to compute $U$.
	\item $\Time_{Sampl}$ be the running time to quantum sample $f$, {\ie} to construct the state $\sum_{\ev \in \Z_q^m} f(\ev) \ket{\ev}$. 
\end{itemize}
\begin{enumerate}
	\item The initialization step of the algorithm can be done as follows
	$$ \sum_{\sv \in \Z_q^n} \ket{\sv} \otimes \sum_{\ev \in \Z_q^m} f(\ev) \ket{\ev} \overset{\textcircled{\scalebox{0.7}{1}}}{\mathlarger{\mathlarger{\rightarrow}}}  \sum_{\substack{\sv \in \Z_q^n \\ \ev \in \Z_q^m}} f(\ev)\ket{\sv}\ket{\Am\T \sv + \ev} = \sum_{\sv \in \Z_q^n} \ket{\sv}\ket{\psi_{\sv}},$$
	which corresponds to the initial state by adding a $\ket{\zerov}$ register and reordering. In $\textcircled{\scalebox{0.7}{1}}$, we use the fact that $\sv \rightarrow \Am\T\sv$ is easily computable and apply this operation coherently. We then need to compute $\sum_{\ev \in \Z_q^m} f(\ev) \ket{\ev}$ which takes some time $\Time_{Sampl}$. In practice, $f$ is chosen such that this state can be computed efficiently. The running time of this is therefore in $O(\Time_{Sampl} + \poly(m,\log(q)))$.
	\item In step $3$, before the measurement, we have the state 
	\begin{align*}
		\frac{1}{\sqrt{q^n}} 	& \sum_{\sv,\sv' \in \Z_q^n} \omega^{-\yv \cdot \sv} \gamma_{\sv,\sv'} \ket{\wpsi_{\sv,\sv'}} \ket{\sv'}\ket{\sv - \sv'} = \\ 
		~											& \frac{1}{\sqrt{q^n}}\left(\sum_{\sv}  \omega^{-\yv \cdot \sv} \sqrt{\Pdec}  \ket{\wpsi_{\sv,\sv}}\ket{\sv}\ket{\zerov} + \sum_{\sv,\sv' \neq \sv}  \omega^{-\yv \cdot \sv} \sqrt{1-\Pdec}  \ket{\wpsi_{\sv,\sv'}}\ket{\sv'}\ket{\sv - \sv'}\right).
	\end{align*}
	which means that we successfully measure $\zerov$ in the last register with probability $\Pdec$ and that conditioned on this outcome, the resulting state is the following: $\frac{1}{\sqrt{q^n}}\sum_{\sv \in \Z_q^n}  \omega^{-\yv \cdot \sv}   \ket{\wpsi_{\sv,\sv}}\ket{\sv}\ket{0}$. We therefore have to repeat steps $1$ to $3$ $O(\frac{1}{\Pdec})$ times. Moreover, step $2$ requires to compute $U$, which takes time $\Time_{U}$ which is the running time of the $\SLWE$ algorithm. From there, we conclude that the time required for this algorithm to successfully pass step $3$ is 
	$$O\left(\frac{1}{\Pdec} \left(\Time_{\SLWE} + \Time_{Sampl} + \poly(m,\log(q))\right)\right).$$ 
	\item In order {to obtain $\ket{\phi_\yv}$ in the end of step $4$}, we start from $U \ket{\psi_\sv}\ket{0} = \sum_{\sv' \in \Z_q^n} \gamma_{\sv,\sv'} \ket{\wpsi_{\sv,\sv'}}\ket{\sv'}$ which implies 
  $$ \bra{\psi_{\sv}}\bra{\zerov} \cdot  U^\dagger  \left(\ket{\wpsi_{\sv,\sv}}\ket{\sv}\right) = \bra{\wpsi_{\sv,\sv}}\bra{\sv} \cdot U \left(\ket{\psi_{\sv}}\ket{\zerov}\right) = \gamma_{\sv,\sv} = \sqrt{\Pdec}.$$
	This means that for each $\sv \in \Z_q^n$,  we can indeed write 
	$$ U^\dagger(\ket{\wpsi_{\sv,\sv}}\ket{\sv}) = \sqrt{\Pdec} \ket{\psi_{\sv}}\ket{0} + \sqrt{1 - \Pdec} \ket{Z_\sv},$$
	for some unit vector $\ket{Z_{\sv}}$ orthogonal to $\ket{\psi_{\sv}}\ket{0}$, which justifies step $4$ of the algorithm. Finally, in step $5$, we have to perform $m$ quantum Fourier transforms in $\Z_q$ and measure, which takes time $\poly(m,\log(q))$.
\end{enumerate}
From this analysis, we can conclude that the total running time of the algorithm satisfies 
$$\Time_{\ISIS} = O\left(\frac{1}{\Pdec} \left(\Time_{\SLWE} + \Time_{Sampl}+ \poly(m,\log(q))\right) \right).$$
In particular, if $\Time_{\SLWE},\Time_{Sampl},\frac{1}{p} = \poly(m,\log(q))$ then $\Time_{\ISIS} = \poly(m,\log(q))$.
The trickier part will be to argue about the success probability of the algorithm, which is the goal of the following section.

\subsection{Proof of the main theorem}
\begin{proof}[of Theorem~\ref{Theorem:1}]
	We start from a quantum algorithm for $\SLWE(\Am,f)$ and we consider the algorithm described in \Cref{Section:Algorithm}.
	The running time of the algorithm has been discussed in the previous section so we just need to prove the success probability.
	We fix $\yv \in \Z_q^{n}$, and let $p'_{\yv}$ be the probability that the algorithm outputs an element $\xv \in \Lambda^\bot_{\yv}(\Am) \cap T$ given this $\yv$.
	
	Consider the state $\ket{Z_\sv}$ defined in step $4$ of the Algorithm. We write 
	$$ \ket{Z_\sv} = \ket{Z^0_\sv}\ket{\zerov} + \sum_{\vv \neq \zerov} \ket{Z^\vv_\sv}\ket{\vv},$$
	with in particular $\norm{\ket{Z^0_\sv}} \le 1$. Recall that $\ket{Z_\sv}$ is orthogonal to $\ket{\psi_\sv}\ket{\zerov}$ which implies that $\ket{Z^0_\sv}$ is orthogonal to $\ket{\psi_\sv}$. At step $5$ of the algorithm before the final measurement, we  have the state 
	$$ \ket{\Omega^\yv} = \frac{1}{\sqrt{q^n}} \sum_{\sv \in \Z_q^n} \left(\omega^{-\yv \cdot \sv} \sqrt{p}\ket{\widehat{\psi_\sv}}\ket{\zerov} + \omega^{-\yv \cdot \sv}\sqrt{1-p}\ket{\widehat{Z^0_\sv}}\ket{\zerov} + \sum_{\vv \neq \zerov}\ket{\widehat{Z_\sv^{\vv}}}\ket{\vv}\right)$$
	Now, let us define 
	$$ \ket{\Xi_0^\yv} \eqdef \frac{1}{\sqrt{q^n}} \sum_{\sv \in \Z_q^n} \left(\omega^{-\yv \cdot \sv} \sqrt{p}\ket{\widehat{\psi_\sv}}+ \omega^{-\yv \cdot \sv}\sqrt{1-p}\ket{\widehat{Z^0_\sv}}\right),$$
	so that $\ket{\Omega^\yv} = \ket{\Xi_0^\yv} \ket{\zerov} + \sum_{\vv \neq \zerov}\ket{\widehat{Z_\sv^{\vv}}}\ket{\vv}.
	$
	Since we only measure the first register and we succeed when we have an element of $\Lambda^\bot_\yv(\Am) \cap T$, we have
	$
		p'_\yv \ge \sum_{\xv \in \Lambda^\bot_\yv(\Am) \cap T} |\braket{\xv}{\Xi_\zerov^\yv}|^2.
	$
	  We now define the $\{z_{\sv,\xv}\}$ such that $\ket{\widehat{Z^0_\sv}} = \sum_{\xv \in \Z_q^m} z_{\sv,\xv} \ket{\xv}$,  and have the following lemma.

	\begin{lemma} For each $\yv \in \Z_q^n$, 
		$$p'_{\yv} \ge \sum_{\xv \in \Lambda^\bot_{\yv}(\Am) \cap T} \left|\sqrt{\Pdec} \sqrt{q^n} \hf(\xv) + \frac{\sqrt{1 - \Pdec}}{\sqrt{q^n}} \sum_{\sv \in \Z_q^n} \omega^{- \sv \cdot \yv} z_{\sv,\xv} \right|^2.$$
	\end{lemma}
	\begin{proof}
		In order to compute $p'_{\yv}$, we have to compute $\ket{\Xi_0}$.
		Using \Cref{Proposition:PsiFourier}, {and the identity $\sum_{\yv' \in \Z_q^n} \omega^{(\yv' -\yv) \cdot \sv} = 0$ if $\yv' \neq \yv$, and $q^n$ otherwise,} we have
		$$ \sum_{\sv \in \Z_q^n} \omega^{-\yv \cdot \sv} \ket{\widehat{\psi_{\sv}}} = \sum_{\sv \in \Z_q^n} \sum_{\yv' \in \Z_q^n} \omega^{(\yv' -\yv) \cdot \sv}  \sum_{\xv \in \Lambda^{\bot}_{\yv'}(\Am)}  \hf(\xv) \ket{\xv} = q^n \sum_{\xv \in \C^\bot_{\yv}(\Am)} \hf(\xv) \ket{\xv}.$$
		From there, we have 
		$$ \ket{\Xi_0} = \sqrt{q^n p} \sum_{\xv \in \C^\bot_{\yv}(\Am)} \hf(\xv)\ket{\xv} + \sqrt{\frac{1-\Pdec}{q^n}}\sum_{\xv \in \Z_q^m} \sum_{\sv \in \Z_q^n} \omega^{-\yv \cdot \sv} z_{\sv,\xv}\ket{\xv}. $$	
		The algorithm computes this state and measures in the computational basis. The probability to output an element of $\C^\bot_{\yv}(\Am) \cap T$ is therefore 
		$$ p'_{\yv} \ge \sum_{\xv \in \Lambda^\bot_\yv(\Am) \cap T} |\braket{\xv}{\Xi_\zerov^\yv}|^2 =  \sum_{\xv \in \C^\bot_{\yv}(\Am) \cap T} \left|\sqrt{q^n{\Pdec}}\hf(\xv) + \sqrt{\frac{1-\Pdec}{q^n}}\sum_{\sv \in \Z_q^n} \omega^{-\yv \cdot \sv} z_{\sv,\xv}\right|^2. $$
	\end{proof}
	We can now go continue the proof of our theorem. We write 
		\begin{align*}
			p'_{\yv} & \ge \sum_{\xv \in \C^\bot_{\yv}(\Am) \cap T} \left|\sqrt{\Pdec} \sqrt{q^n} \hf(\xv) + \frac{\sqrt{1 - \Pdec}}{\sqrt{q^n}} \sum_{\sv \in \Z_q^n} \omega^{-\yv \cdot \sv} z_{\sv,\xv} \right|^2 \\
			& = \sum_{\xv \in \C^\bot_{\yv}(\Am) \cap T} \left[{\Pdec}{q^n} |\hf(\xv)|^2 + \frac{{1 - \Pdec}}{{q^n}} \left|\sum_{\sv \in \Z_q^n} \omega^{-\yv \cdot \sv} \hzsy \right|^2\right.\\
			&  + \left.2Re\left(\sqrt{\Pdec} \sqrt{q^n} \hf(\xv)\frac{\sqrt{1 - \Pdec}}{\sqrt{q^n}} \sum_{\sv \in \Z_q^n} \omega^{\yv \cdot \sv} \ohzsy\right)\right]
		\end{align*}
		where we used $|a+b|^2 = (a+b)(\overline{a} + \overline{b}) = |a|^2 + |b|^2 + 2Re(a\overline{b})$. We now bound each term separately. We first write 
		\begin{align*}
			\E_{\yv \Unif \Z_q^n} \left[\sum_{\xv \in \C^\bot_{\yv}(\Am) \cap T} \left({\Pdec}{q^n} |\hf(\xv)|^2\right)\right] & = \Pdec (1-\eta) \\
			\forall \yv \in \Z_q^n, \ \sum_{\xv \in \C^\bot_{\yv}(\Am) \cap T} \left(\frac{{1 - \Pdec}}{{q^n}} \left|\sum_{\sv \in \Z_q^n} \omega^{-\yv \cdot \sv}\hzsy \right|^2\right) & \ge 0 \\
			\sum_{\xv \in \C^\bot_{\yv}(\Am) \cap T} \left(2Re\left(\sqrt{\Pdec} \sqrt{q^n} \hf(\xv)\frac{\sqrt{1 - \Pdec}}{\sqrt{q^n}} \sum_{\sv \in \Z_q^n} \omega^{\yv \cdot \sv}\ohzsy\right)\right) &=\\
			2\sqrt{\Pdec(1-\Pdec)} Re\left(\sum_{\xv \in \C^\bot_{\yv}(\Am) \cap T, \sv \in \Z_q^n} \hf(\xv)\omega^{\yv \cdot \sv}\ohzsy\right)
		\end{align*}

		From there, we write 
		\begin{align*} \E_{\yv \Unif \Z_q^n} \left[p'_{\yv}\right] & \ge \Pdec(1-\eta) + 2\sqrt{\Pdec(1-\Pdec)} \E_{\yv \Unif \Z_q^n} \left[Re\left(\sum_{\substack{\xv \in \C^\bot_{\uv}(\Am) \cap T \\ \sv \in \Z_q^n}} \hf(\xv)\omega^{\yv \cdot \sv} \ohzsy\right)\right] \end{align*}
		In order to conclude, prove the following lemma
		\begin{lemma}
			$$ \forall \sv \in \Z_q^n, \ \left|\sum_{\yv \in \Z_q^n}\sum_{{\xv \in \C^\bot_{\yv}(\Am) \cap T }} \hf(\xv)\omega^{\yv \cdot \sv}\ohzsy\right| \le \sqrt{\eta}.$$
		\end{lemma}
		\begin{proof}
		We start from the equality $\braket{Z^0_\sv}{\psi_\sv} = \braket{\widehat{Z^0_\sv}}{\widehat{\psi_\sv}} = 0$ for each $\sv \in \Z_q^n$, which can be rewritten
			\begin{align*}
			\forall \sv \in \Z_q^n, \ \sum_{\yv  \in  \Z_q^n} \sum_{\xv \in \C^\bot_{\yv}(\Am)} \omega^{\yv \cdot \sv} \hf(\yv)\ohzsy = 0 
			\end{align*}
			This implies that for each $\sv \in \Z_q^n$,
			\begin{align*}
				\left|\sum_{\yv \in \Z_q^n}\sum_{{\xv \in \C^\bot_{\yv}(\Am) \cap T }} \hf(\xv)\omega^{\yv \cdot \sv}\ohzsy\right| & = \left|\sum_{\yv \in \Z_q^n}\sum_{{\xv \in \C^\bot_{\yv}(\Am) \cap \overline{T} }} \hf(\xv)\omega^{\yv \cdot \sv}\ohzsy\right| \\
				& =  \left|\sum_{{\xv \in  \overline{T} }} \hf(\xv)\omega^{\Am\T \xv \cdot \sv}\ohzsy\right| \\
				& \le \sqrt{\sum_{\xv \in  \overline{T} } |\hf(\xv)\omega^{\Am\T \xv \cdot \sv}|^2}\sqrt{\sum_{\yv \in \overline{T}}\ohzsy} \\
				& \le \sqrt{\eta}\sqrt{1} = \sqrt{\eta},
			\end{align*}
			where we used the fact that the $\ket{\widehat{Z^0_\sv}}$ have norm at most $1$.
\end{proof}
We can now conclude our main proof. We have
\begin{align*}
  \E_{\yv \Unif \Z_q^n} \left[Re\left(\sum_{\substack{\xv \in \C^\bot_{\uv}(\Am) \cap T \\ \sv \in \Z_q^n}} \hf(\xv)\omega^{\yv \cdot \sv} \ohzsy\right)\right] & = \frac{1}{q^n} \sum_{\sv \in \Z_q^n} Re\left(\sum_{\yv \in \Z_q^n}\sum_{\xv \in \Lambda^\bot_\yv(\Am)} \hf(\xv)\omega^{\yv \cdot \sv}\ohzsy\right) \\
	& \ge - \frac{1}{q^n} \sum_{\sv \in \Z_q^n} \left|\sum_{\yv \in \Z_q^n}\sum_{\xv \in \Lambda^\bot_\yv(\Am)} \hf(\xv)\omega^{\yv \cdot \sv}\ohzsy\right| \\
	& \ge - \frac{1}{q^n} \sum_{\sv \in \Z_q^n} \sqrt{\eta} \\
	& = - \sqrt{\eta}
\end{align*}
{Plugging this lower bound in the expression above, we obtain
$$
  \E_{\yv \Unif \Z_q^n} \left[p'_{\yv}\right] \ge \Pdec(1-\eta) - 2\sqrt{\Pdec(1-\Pdec)\eta}
$$
which concludes the proof.}
\end{proof}

We present the instantiation with Gaussian functions (Proposition~\ref{Proposition:Gaussians}) in Appendix~\ref{Appendix:Gaussians}

\COMMENT{\subsection{Gaussian distributions}

For a fixed $q$, we define the discrete Gaussian distribution $\chi_\sigma$ with parameter $\sigma$ on $\Z_q$ as follows
$$\chi_\sigma(x) = \frac{\sum_{k \in \mathbb{Z}} e^{-\pi((x + kq)^2/\sigma^2)}}{\sum_{x \in \Z_q}\sum_{k \in \mathbb{Z}} e^{-\pi((x + kq)^2/\sigma^2)}} , \quad \text{for } x \in \Z_q.$$

We have the following claims
\begin{claim}
	We have 
	$ \widehat{\chi_\sigma}(x) \eqdef \chi_{q/\sigma}(x) + e^{-\Theta(q^2/\sigma^2)},$
	where the Fourier transform is on $\Z_q$.
\end{claim}
\begin{claim}
	Let $T_\tau = \{\xv \in \Z_q^m : \norm{\xv}_2 \le \tau\}$. For $\sigma \ll q$, we have 
	$$\sum_{\xv \in T_\tau} |\chi_{q/\sigma}^m(\xv)|^2 = 1 - o(1),$$
	for $\tau = O(\frac{q\sqrt{m}}{\sigma})$.
\end{claim} 

\begin{proposition}
	Assume we have an efficient quantum algorithm for $\SLWE_{q,m,n,\chi_\sigma^m}$ with high probability, then we have an efficient quantum algorithm for $\ISIS(q,m,n,\beta = O(\frac{\sqrt{m}q}{\sigma}))$.
\end{proposition}
\begin{proof}
	This is an immediate corollary of \Cref{Theorem:1} and the above claims. \andre{TODO}
\end{proof}
}
\section{Reverse direction $\SLWE \rightarrow \ICLWE$}

\begin{definition}\label{Definition:Clean}
	We say that a unitary $U$ for $\ICLWE(\Am,f)$ has fidelity $\gamma$ if for all $\yv \in \ZZ_q^n$,
	$$
		U\ket{\yv}\ket{\zerov} = \ket{\yv}\ket{\wphi_\yv} \quad \text{with } \E_{\yv \Unif \Z_q^n}\left[\left|\braket{\wphi_\yv}{W_\yv}\right|\right] = \gamma,
	$$
	for some unit vectors $\ket{\wphi_\yv}$; where $\ket{W_\yv} = \frac{1}{\sqrt{w_\yv}}\sum_{\sv \in \Z_q^n} \omega^{-\yv \cdot \sv}\sum_{\ev \in \Z_q^m} f(\ev)\ket{\Am\T \sv + \ev}$.
\end{definition}

We prove the following theorem

\begin{theorem}\label{Theorem:2}
	Let $q,m,n$ be positive integers, $f$ be a function with domain $\Z_q^m$ and codomain $\mathbb{C}$, and $\Am \in \Z_q^{n \times m}$. If:
	\begin{itemize}
		\item We have an efficient unitary for $\ICLWE(\Am,f)$ with fidelity $\gamma = 1 - \eps'$.
		\item $\E_{\yv}\left[\sqrt{\frac{w_\yv}{q^n}}\right] = 1-\eps$, where $w_\yv$ is the normalization factor so that  $$\ket{W_\yv} = \frac{1}{\sqrt{w_\yv}}\sum_{\sv \in \Z_q^n} \omega^{- \sv \cdot \yv}\sum_{\ev \in \Z_q^m} f(\ev) \ket{\Am\T \sv + \ev} \quad \text{is a unit vector.}$$
	\end{itemize}
	Then one can construct an efficient quantum algorithm that solves $\SLWE(\Am,f)$ that succeeds with probability $\left(1 - \eps - \eps' - 2\sqrt{\eps\eps'}\right)^2$.
\end{theorem}

\begin{proof}
  We prove this theorem below in the case where $\varepsilon = \varepsilon' = 0$, and defer the proof for general $\varepsilon$ and $\varepsilon'$ to \Cref{Appendix:Reverse1}.

  Remember that $\ket{\widehat{\psi_\sv}} = \frac{1}{q^n}\sum_{\yv \in \Z_q^n} \omega^{\yv \cdot \sv} \sqrt{w_\yv}\ket{\widehat{W_\yv}}$ (\cref{Proposition:PsiToW}).
  First apply a $\QFT$ to $\ket{\psi_\sv}$, then apply the function $\xv \mapsto \Am \xv$ in superposition and store the result in a second register.
  As each $\ket{\widehat{W_\yv}}$ is a superposition of vectors in $\Lambda^\perp_\yv(\Am)$, this yields
  $$
    \frac{1}{q^n}\sum_{\yv \in \Z_q^n} \omega^{\yv \cdot \sv} \sqrt{w_\yv}\ket{\widehat{W_\yv}} \ket{\yv}
  $$

  Now, the existence of an efficient unitary for $\ICLWE(\Am, f)$ implies the existence of another efficient unitary $U$ mapping $\ket{\yv}\ket{\zerov}$ to $\ket{\yv}\ket{\widehat{W_\yv}}$.
  Swapping the registers and applying $U^\dagger$ then zeroizes the $\ket{\widehat{W_\yv}}$ register, allowing us to remove it and leaving only
  $$
    \ket{\phi} \eqdef \frac{1}{q^n}\sum_{\yv \in \Z_q^n} \omega^{\yv \cdot \sv} \sqrt{w_\yv} \ket{\yv}
  $$
  Observe finally that the overlap between $\ket{\phi}$ and $\frac{1}{q^n}\sum_{\yv \in \Z_q^n} \omega^{\yv \cdot \sv} \ket{\yv}$ is $\E_\yv \frac{\sqrt{w_\yv}}{q^n} = 1$.
  Thus, applying an inverse $\QFT$ on $\ket{\phi}$, and then measuring yields $\sv$ with probability $1$.
\end{proof}

\section{Conditional reverse reduction $\ICLWE \rightarrow \ISIS$}
In this last section, we provide another reduction.
We show that if we have an algorithm for $\ISIS(\Am,f)$ which has a specific form then it can be used to solve $\SLWE(\Am,f)$. We start by introducing the definition of a randomness-recoverable algorithm.
This is a random algorithm whose random tape value can be recovered given the corresponding output.
\begin{definition}[Randomness-Recoverable Algorithm for $\ISIS$]
  \label{def:rdrecov}
	Let $\aa$ be an algorithm for $\ISIS(\Am,T)$, and denote by $\aa(\yv; r)$ the output of $\aa$ on input $\yv$ using random tape $r \in \bin^{\ell}$.
  $\aa$ is said to be perfectly randomness-recoverable if there is an algorithm $\rdext$ satisfying
  $$
  \forall r \in \zo^\ell, \ \forall \yv \in \Z_q^n, \ \rdext(\yv,\aa(\yv; r)) = r.
  $$
\end{definition}

\begin{theorem}
  \label{th:iclwe-red-to-isis}
  Let $\aa$ be an algorithm for $\ISIS(\Am, T)$ with time complexity $t$.
  Assume the following properties:
  \begin{itemize}
    \item $\aa$ is perfectly randomness-recoverable;
    \item $\aa$ is $\eps$-close to being solution-uniform {\ie} if we define $p_\yv(\xv) = \Pr_{\rv \Unif \zo^\ell}\left(\aa(\yv,\rv) = \xv\right)$, we have 
    $$ \Delta(p_\yv,u_\yv) = \eps_\yv \quad \text{and} \quad \E_{\yv \Unif \Z_q^n}\left[\eps_\yv\right] = \eps,$$
    where $\Delta(p_\yv,u_\yv)$ is the statistical distance between the distribution $p_\yv$ and the probability function $u_\yv = \frac{1}{|T \cap \Lambda^\bot_{\yv}(\Am)|} \one_{T \cap \Lambda^\bot_{\yv}(\Am)}$.
  \end{itemize}
  Then there exists an efficient algorithm solving $\ICLWE(\Am, f)$ with $\hf = \one_T$ with fidelity $1-\eps$ and time complexity $\poly(t)$.
\end{theorem}

\begin{remark}
  We do not explicitly say anything about the success probability of the algorithm $\aa$. However, the fact that it is $\eps$-close to being solution-uniform implies that with a uniformly random choice of randomness $r$, the algorithm outputs a valid solution ({\ie} $\in T \cap \Lambda^\bot_\yv(\Am)$) with probability at least $1-\eps$ on average on $\yv$.
\end{remark}

\begin{proof}
  We show that there is an efficient process mapping $\ket{\yv}\ket{0}$ to $\ket{\yv}\ket{W_\yv}$ for all $\yv \in \Z^n_q$.
  Start with $\ket{y}$ as first register, then prepare a uniform superposition of $r \in \bin^\ell$ in the second register:
  $$
    \frac{1}{\sqrt{2^\ell}}\ket{\yv} \sum_{r \in \bin^\ell}\ket{r}
  $$
  Apply $\aa$ in superposition over the first and second registers and store the result in a third register:
  $$
    \frac{1}{\sqrt{2^\ell}}\ket{\yv} \sum_{r \in \bin^\ell}\ket{r} \ket{\aa(\yv; r)}
  $$
 From $\rdext$, we have access to the quantum unitary  $\Ugate_\rdext$ mapping $\ket{\yv}\ket{\aa(y;r)}\ket{0}$ to $\ket{\yv}\ket{\aa(y;r)}\ket{r}$ for any $r \in \zo^l$. We then perform the following operations.
  \begin{align*}
    \frac{1}{\sqrt{2^\ell}}\ket{\yv} \sum_{r \in \bin^\ell}\ket{r} \ket{\aa(\yv; r)}
    \xrightarrow{SWAP} & \frac{1}{\sqrt{2^\ell}} \ket{\yv} \sum_{r \in \bin^\ell}\ket{\aa(\yv; r)} \ket{r}\\
    \xrightarrow{\Ugate^\dagger_\rdext} & \frac{1}{\sqrt{2^\ell}} \ket{\yv} \sum_{r \in \bin^\ell}\ket{\aa(\yv; r)} \ket{0}\\
    \xrightarrow{discard} & \frac{1}{\sqrt{2^\ell}} \ket{\yv} \sum_{r \in \bin^\ell}\ket{\aa(\yv; r)} \\
    \xrightarrow{inverse \ \QFTw} & \frac{1}{\sqrt{2^\ell}} \ket{\yv} \otimes  \QFTw^{-1}_{\Z_q^n}\left(\sum_{r \in \bin^\ell}\ket{\aa(\yv; r)}\right)
  \end{align*}
  For this final step, one has to be careful because we do not necessarily restrict the outputs of $\aa(\yv;r)$ to elements of $\Z_q^n$. For example, we will consider algorithms which sometimes outputs Abort. We extend the operation $\QFTw^{-1}_{\Z_q^n}$ so that it applies the identity to elements outside of $\Z_q^n$. 
  
  Let $\ket{W'_\yv} = \frac{1}{\sqrt{2^{\ell}}} \QFTw^{-1}_{\Z_q^n}\left(\sum_{r \in \bin^\ell}\ket{\aa(\yv; r)}\right)$. In order to show that we have a quantum algorithm that solves $\ICLWE(\Am,f)$ with fidelity $1-\eps$, we need to compute the inner products $|\braket{W_\yv}{W'_\yv}| = |\braket{\widehat{W_\yv}}{\widehat{W'_\yv}}|$. Since we have $\hf \sim \one_T$, we know that 
  \begin{align*}
  \ket{\widehat{W_\yv}} & = \frac{1}{\sqrt{|\Lambda^\bot_\yv(\Am) \cap T|}}\sum_{\xv \in \Lambda^\bot_\yv(\Am) \cap T} \ket{\xv} \\
  \ket{\widehat{W'_\yv}} & = \frac{1}{\sqrt{2^{\ell}}} \sum_{r \in \bin^\ell}\ket{\aa(\yv; r)}
  \end{align*}
  
  We then obtain 
  \begin{align*}
  	|\braket{\widehat{W_\yv}}{\widehat{W'_\yv}}| = \sum_{\xv} \sqrt{p_\yv(\xv) u_\yv(\xv)} = F(p_\yv,u_\yv),
  \end{align*}

   Where $F$ is the fidelity between the two probability functions. By the Fuchs-van de Graaf inequality~\cite{FvdG99}, we have $F(p_{\yv},u_{\yv}) \ge 1 - \Delta(p_{\yv},u_{\yv}) = 1 - \eps_\yv$. From there, we can conclude that this algorithm solves $\ICLWE(\Am,\one_{T})$ with fidelity at least $\E_{\yv}\left[1 - \eps_\yv\right] = 1 - \eps$. 
\end{proof}

In our work, we only require the above theorem for algorithms which are perfectly randomness recoverable.
For sake of completeness, we extend the theorem to setting where the algorithm is almost randomness recoverable in in \Cref{Appendix:ImperfectRR}.

\section{Specific instantiation of the reverse reduction}
\label{sec:instantiations}

In this section, we describe a new algorithm to solve $\ISIS$ in a certain range of parameters, and show that, by instantiating our reduction with this solver, we can recover a recent result from \cite{BJK+25}, proving that $\SLWE$ is solvable in super-polynomial time when the modulo is a power of two.
Interestingly, while \cite{BJK+25} introduced a somehow complex quantum algorithm for this task, involving several layers of preparing, measuring, and combining quantum states, our method is arguably much more straightforward, and ``only'' involves $\QFT$ gates and coherent implementation of classical algorithms.
We believe that this hints that the hardness of $\SLWE$ mostly lies in the hardness of the underlying $\ISIS$ problem.

In this section, we sometimes abuse the notations and write, for $\xf \in \ZZ_q^m$, $\xf \in \ZZ_2^m$ to denote that the coefficients of $\xf$ all belong $\bin$.

\subsection{An algorithm for solving $\ISIS$ with non-trivial $\ell_\infty$-norm bounds}
We first recall one of the results of~\cite{BJK+25}.
\begin{proposition}
  For some parameters $n,m,q = 2^\ell$ with $q = \poly(n)$ and $m = 2^{O(\log(n)\log(q))}$, there exists an algorithm for $\SLWE(\Am,f)$ in the case $\Am$ is randomly chosen in $\Z_q^{n \times m}$ and $f$ is the function such that $\hf = \one_{\Z_2^m}$.
\end{proposition}
In this section, we describe an efficient $\ISIS$ solver algorithm, and show that we can recover the above proposition simply by plugging this algorithm in our reverse reduction.

We start by describing our solver, adapted from the $\SIS$-solver algorithm described in \cite[Appendix A]{CLZ22}, in a generic way.
Recall that instantiating our reduction with a solver requires the solver to be randomness-recoverable.
To prove this property, we need to provide more details on how to implement this algorithm --- in particular how we sample the vectors.
We simply present in this subsection a generic version of the solver, without the aforementioned detailed sampling procedures, and prove its correctness and time-complexity.
A complete version of the algorithm, with detailed sampling procedures, is presented them in \Cref{sec:random-recov-solver}.
\begin{algobox}{$\isissolver$}{isis-solver-mod2-simple}
  \begin{itemize}[leftmargin=0pt,itemsep=1em]
    \item[] \textbf{Parameters:}
    Positive integers $\ell$, $n$, and $w$ such that $w > cn$ for some constant $c > 1$.

    \item[] \textbf{Notations:}
    Let $q = 2^\ell$, and $m = w^\ell$.
    If $\ell > 1$, let $q' = q / 2$, and $m' = m / w$.

    \item[] \textbf{Inputs:} $\Am \in \Z^{n \times m}_q$, $\yv \in \Z^n_q$.

    \item[] \textbf{Output:} $\xf \in \Z^m_2$ such that $\Am \xf = \yv \mod q$, and $\xf \in \ZZ_2^m$.

    \item[] \textbf{Execution of the algorithm (for $\ell = 1$):}
    If $\Am$ is not full-rank, abort (\ie{} return error symbol $\bot$).
    Otherwise, sample a solution $\xv$ to the system $\Am \xv = \yv \mod 2$, and return it.

    \item[] \textbf{Execution of the algorithm (for $\ell > 1$):}
    \begin{enumerate}
      \item \textit{(Sample $\yv$'s shares.)}
      For $i \in \Iint{1}{m'-1}$, sample $\yv_i \Unif \Z_2^n$.
      Let $\yv_{m'} \eqdef \yv  - \sum_{i = 1}^{m' - 1} \yv_i \mod 2$.

      \item \textit{(Build block-matrices.)}
      Let $\Am_1, \dots, \Am_{m'}$ be the matrices in $\Z^{n \times w}_{2}$ such that $\Am = \left[\Am_1 \vert \dots \vert \Am_{m'}\right] \mod 2$.
      If there exists $i \in \Iint{1}{m}$ such that $\Am_i$ is not full rank modulo $2$, abort (\ie{} return error symbol $\bot$).

      \item\textit{(Sample solution pairs for each block.)}
      \label{it:sample-sols-mod2-simple}
      For all $i \in \Iint{1}{m'}$, compute two solutions $\xv^{(1)}_i$ and $\xv^{(2)}_i$ to the equation $\Am_i \xv = \yv_i \mod 2$, such that $\xv^{(1)}_i \neq \xv^{(2)}_i$.
      Then, let $\xv_i \eqdef \xv^{(1)}_i$, and $\zv_i \eqdef \xv^{(2)}_i - \xv^{(1)}_i \mod q$~\footnote{%
        We emphasize that the $\zv_i$'s are defined modulo $q$ and not modulo $2$.
        This turns out to be crucial to enforce the norm constraint.}.

      \item \textit{(Merge block solutions.)}
      Let $\xv \in \Z_2^m$, and $\Zv \in \Z^{m \times m'}_q$ defined as follows:
      \begin{align*}
        \xv \eqdef \left(\begin{array}{c}
          \xv_1\\
          \horzbar\\
          \vdots\\
          \horzbar\\
          \xv_{m'}
        \end{array}\right) & ~ &
        \Zv \eqdef \left[\begin{array}{cccc}
          \zv_1 & ~ & ~ & ~\\
          ~ & \zv_2 & ~ & ~\\
          ~ & ~ & \ddots & ~\\
          ~ & ~ & ~ & \zv_{m'}\\
        \end{array}\right]
      \end{align*} where the empty spots in $\Zm$ are zeroes.
      Note that $\Am \xv = \yv \mod 2$, and $\Am \Zm = \sum_i \Am_i \zv_i = \zerov \mod 2$.

      \item \textit{(Apply next step of the recursion.)}
      Let $\yv' \eqdef \frac{\yv - \Am \xv}{2} \mod q' \in \Z_{q'}^{n}$ and $\Am' \eqdef \frac{\Am\Zm}{2} \mod q' \in \Z_{q'}^{n \times m'}$~\footnote{%
        $\yv'$ and $\Am'$ are well defined, as $\Am \xv = \yv \mod 2$, and $\Am \Zv = 0 \mod 2$.}.
      Run $\isissolver(\Am',\yv')$ with parameters $(\ell-1, n)$, and let $\xv'$ denote the outcome.
      If $\xv = \bot$, abort (\ie{} return error symbol $\bot$).
      Otherwise, return $\xf = \xv + \Zm \xv' \mod q$.
    \end{enumerate}
  \end{itemize}
\end{algobox}

We now prove the correctness and time-complexity of \Cref{algo:isis-solver-mod2-simple}.
We start by proving that \Cref{algo:isis-solver-mod2-simple} only aborts with negligible probability over random matrix $\Am$ and its internal randomness.
\begin{lemma}
  \label{lem:no-abort-mod2}
  Let $n$, $\ell$, and $w$ be positive integers such that $w > cn$ for some constant $c > 1$, and let $q = 2^\ell$.
  For all $\yv \in \Z_q^n$,
  $$
    \probsublong{\Am \sample \Z_q^{n \times m}}{\isissolver(\Am, \yv) = \bot} = \negl(n)
  $$
\end{lemma}

\begin{proof}
  \Cref{algo:isis-solver-mod2-simple} aborts if any block-matrix is not full-rank.
  As $w$ is sufficiently larger than $n$, any block-matrix taken alone is full-rank with probability $1 - \negl(n)$.
  As one run of the algorithm uses polynomially-many block-matrices, showing that they are all uniformly random and independent is sufficient to prove the proposition by a union bound.

  Formally, we prove that when the algorithm is run at level $\ell$ with a uniformly random matrix $\Am \in \ZZ_q^{n \times m}$ as input , then the next-level's matrix $\Am' \in \ZZ_{q'}^{n \times m'}$ produced by the algorithm is also uniformly random, and independent from level-$\ell$'s block-matrices $(\Am_i \mod 2)_i$.

  Write first $\Am = \Am_{0} + 2 \Bm$, where $\Am_{0}$ is over $\Z_2$, $\Bm$ is over $\Z_{q'}$, and the addition and scalar-multiplication are over $\Z_q$.
  Note that, as $\Am$ is uniformly random, and $2$ divides $q$, then $\Am_{0}$ and $\Bm$ are independent and both uniformly random.
  For each block $i$, the corresponding block-matrix is then $(\Am_i \mod 2) = \Am_{0, i}$.
  The $i$-th column $\av'_i$ of the next-level matrix $\Am'$ is then defined as
  $$
    \av'_i  = \frac{\Am_i \zv_i}{2} \mod q'\\
    ~       = \frac{\Am_{0, i} \zv_i}{2} + \Bm_{i} \zv_i \mod q'
  $$
  Now recall that $\zv_i$ depends only on $\Am_{0, i}$.
  Hence $\Bm_{i}$ is uniformly random, and independent from $\Am_{0, i}$ and $\zv_i$.
  As $\zv_i$ is non-zero by construction, and each of its coefficients belong to $\{-1, 0, 1\}$, it comes that $\av'_i$ is uniformly random, and independent from $\Am_{0, i}$ --- and more generally independent from $\Am_0$.

  Applying this reasoning recursively completes the proof.
\end{proof}

We can now proceed with the correctness.
\begin{proposition}
  Let $n$, $\ell$, and $w$ be positive integers such that $w > cn$ for some constant $c > 1$, and let $q = 2^\ell$ and $m = w^\ell$.
  Then, for all $\yv \in \ZZ_q^n$:
  $$
    \probsublong{\begin{array}{l}
      \Am \sample \ZZ_q^{n \times m}\\
      \xv \gets \isissolver(\Am, \yv)
    \end{array}}{\begin{array}{c}
      \Am \xv = \yv \mod q\\
      \land\\
      \xv \in \ZZ_2^m
    \end{array}} = 1 - \negl(n)
  $$
\end{proposition}

\begin{proof}
  By \Cref{lem:no-abort-mod2}, we can safely assume that \Cref{algo:isis-solver-mod2-simple} does not abort, as it only changes the probability by a negligible amount.
  We prove the result by induction on $\ell$.
  When $\ell = 1$, the correctness is immediate.
  When $\ell > 1$, let $\xf$ be an output of the algorithm.
  By construction, $\xf = \xv + \Zv \xv' \mod q$ where
  $$\begin{array}{ll}
    \Am' \eqdef \Am\Zv/2 \mod q'            & \ie\ \Am\Zv = 2\Am' \mod q\\
    \Am' \xv' = \yv' \mod q'                & \ie\ 2 \Am' \xv' = 2 \yv' \mod q\\
    \yv' \eqdef (\yv - \Am \xv) / 2 \mod q' & \ie\ \Am \xv = \yv - 2\yv' \mod q
  \end{array}$$
  where the second line follows by induction hypothesis, as $\xv'$ is an output of \Cref{algo:isis-solver-mod2-simple} on input $(\Am', \yv')$.
  Using these identities, it comes $\Am \xf = \Am \xv + \Am \Zv \xv' \mod q = \Am \xv + 2 \Am' \xv' \mod q = \Am \xv + 2 \yv' \mod q = \yv \mod q$ which concludes the first part of the proof.

  It remains to show that $\xv \in \ZZ^m_2$.
  By induction hypothesis, $\xv' \in \ZZ_2^m$.
  Furthermore, the block-column structure of $\Zv$ enforces $\xf_i = \xv_i + \xv'[i] \cdot \zv_i$ for all $i$, where $\xv'[i] \in \bin$ denotes the $i$-th coefficient of $\xv'$~\footnotemark.
  \footnotetext{This is as opposed to the notation $\xv'_i$ that denotes the $i$-th \emph{block} of $\xv'$.}%
  Recall that $\xv_i$ is defined as the first solution $\xv_i^{(1)}$ of the pair $(\xv_i^{(1)}, \xv_i^{(2)})$ sampled in step \ref{it:sample-sols-mod2-simple}, and $\zv_i$ as $\xv_i^{(2)} - \xv_i^{(1)}$, where, crucially, the difference is over $\ZZ_q$.
  Thus, $\xf_i \in \set{\xv_i^{(1)}, \xv_i^{(2)}}$, which implies $\xf \in \ZZ^m_2$, and concludes the proof.
\end{proof}

\begin{proposition}
  \label{prop:time-complexity-mod2}
  Let $n$, $\ell$, and $w$ be positive integers such that $w > cn$ for some constant $c > 1$, and let $q = 2^\ell$.
  Then \Cref{algo:isis-solver-mod2-simple} runs in time $2^{O(\log(n)\log(q))}$.
\end{proposition}

\begin{proof}
  This is straightforward, as $m = w^\ell = 2^{O(\log(n)\log(q))}$, and the algorithm consists in less than $m$ blocks, each of them itself consisting in solving $2$ linear systems of size $\poly(n)$ over $\ZZ_2$, which can be done efficiently using Gaussian elimination.
\end{proof}

\subsection{Randomness recovery}
\label{sec:random-recov-solver}

We show in this subsection that \Cref{algo:isis-solver-mod2-simple} can be made randomness-recoverable.
In order to do that, we present a completed version of this algorithm, in which we explicitely describe how we sample the inhomogeneous and homogeneous solutions at each level.
We write the difference between this completed version and \Cref{algo:isis-solver-mod2-simple} in \textcolor{blue}{blue}.
This allows us to prove that the algorithm, we considering random tape as an input, is almost bijective, and allows to conclude that it is randomness-recoverable.
In the rest of the section, $\isissolver$ denotes this detailed version: \Cref{algo:isis-solver-mod2-detailed}.

\begin{algobox}{$\isissolver$}{isis-solver-mod2-detailed}
  \begin{itemize}[leftmargin=0pt,itemsep=1em]
    \item[] \textbf{Parameters:}
    Positive integers $\ell$, $n$, and $w$ such that $w > cn$ for some constant $c > 1$.

    \item[] \textbf{Subroutine:}
    A \emph{deterministic} $\matext$ procedure that takes a rank-$n$ matrix $\Am \in \ZZ_2^{n \times w}$ as input, as well as an integer $r \in \Iint{n}{w}$, and returns an ``extended'' matrix $\tAm \in \ZZ_2^{r \times w}$ of rank $r$, and such that the first $n$ rows of $\tAm$ exactly form the matrix $\Am$.

    \item[] \textbf{Notations:}
    Let $q = 2^\ell$, and $m = w^\ell$.
    If $\ell > 1$, let $q' = q / 2$, and $m' = m / w$.

    \item[] \textbf{Inputs:} $\Am \in \Z^{n \times m}_q$, $\yv \in \Z^n_q$.

    \item[] \textbf{Output:} $\xf \in \Z^m_2$ such that $\Am \xf = \yv \mod q$, and $\xf \in \ZZ_2^m$.

    \item[] \textbf{Execution of the algorithm (for $\ell = 1$):}
    If $\Am$ is not full-rank, abort (\ie{} return error symbol $\bot$ \textcolor{blue}{and the random tape}).
    \textcolor{blue}{Otherwise, extend $\Am$ to the rank-$w$ matrix $\tAm$ using $\matext$ procedure, then consume $w-n$ bits from the random tape to construct a vector $\uv \in \ZZ_2^{w - n}$, and finally compute and return the only solution $\xv$ to the system $\tAm \xv = (\yv \| \uv)\T \mod 2$.}

    \item[] \textbf{Execution of the algorithm (for $\ell > 1$):}
    \begin{enumerate}
      \item \textit{(Sample $\yv$'s shares.)}
      \textcolor{blue}{For $i \in \Iint{1}{m'-1}$, consume $n$ bits from the random tape to construct $\yv_i \Unif \Z_2^n$.}
      Let $\yv_{m'} \eqdef \yv  - \sum_{i = 1}^{m' - 1} \yv_i \mod 2$.

      \item \textit{(Build block-matrices.)}
      Let $\Am_1, \dots, \Am_{m'}$ be the matrices in $\Z^{n \times w}_{2}$ such that $\Am = \left[\Am_1 \vert \dots \vert \Am_{m'}\right] \mod 2$.
      If there exists $i \in \Iint{1}{m}$ such that $\Am_i$ is not full rank modulo $2$, abort (\ie{} return error symbol $\bot$ \textcolor{blue}{and the random tape}).

      \item\textit{(Sample solution pairs for each block.)}
      \label{step:sol-sample}
      \textcolor{blue}{For all $i \in \Iint{1}{m'}$, extend $\Am_i$ to a rank-$w-1$ matrix $\tAm_i$ using $\matext$ procedure, then consume $w-n-1$ bits from the random tape to construct a vector $\uv_i \in \ZZ_2^{w-n-1}$.
      Compute then the two solutions $\xv^{(1)}$ and $\xv^{(2)}$ to the equation $\tAm_i \xv = (\yv_i \| \uv_i)\T \mod 2$, such that $\xv^{(1)}_i \preccurlyeq \xv^{(2)}_i$ in lexicographical order.}
      Then, let $\xv_i \eqdef \xv^{(1)}_i$, and $\zv_i \eqdef \xv^{(2)}_i - \xv^{(1)}_i \mod q$.

      \item \textit{(Merge block solutions.)}
      \label{step:merge-sols}
      Let $\xv \in \Z_2^m$, and $\Zv \in \Z^{m \times m'}_q$ defined as follows:
      \begin{align*}
        \xv \eqdef \left(\begin{array}{c}
          \xv_1\\
          \horzbar\\
          \vdots\\
          \horzbar\\
          \xv_{m'}
        \end{array}\right) & ~ &
        \Zv \eqdef \left[\begin{array}{cccc}
          \zv_1 & ~ & ~ & ~\\
          ~ & \zv_2 & ~ & ~\\
          ~ & ~ & \ddots & ~\\
          ~ & ~ & ~ & \zv_{m'}
        \end{array}\right]
      \end{align*} where the empty spots in $\Zm$ are zeroes.
      Note that $\Am \xv = \yv \mod 2$, and $\Am \Zm = \sum_i \Am_i \zv_i = \zerov \mod 2$.

      \item \textit{(Apply next step of the recursion.)}
      Let $\yv' \eqdef \frac{\yv - \Am \xv}{2} \mod q' \in \Z_{q'}^{n}$ and $\Am' \eqdef \frac{\Am\Zm}{2} \mod q' \in \Z_{q'}^{n \times m'}$.
      Run $\isissolver(\Am',\yv')$ with parameters $(\ell-1, n)$, and let $\xv'$ denote the outcome.
      If $\xv' = \bot$, abort (\ie{} return error symbol $\bot$ \textcolor{blue}{and the random tape}).
      Otherwise, return $\xf = \xv + \Zm \xv' \mod q$.
    \end{enumerate}
  \end{itemize}
\end{algobox}

\begin{remark}
  We describe an example of an efficient deterministic matrix extension procedure in \Cref{app:matext} to be used as the $\matext$ subroutine.
\end{remark}

\begin{remark}
  Note in addition to the detailed sampling, another change in \Cref{algo:isis-solver-mod2-detailed} is that it now returns its random tape when aborting.
  An advantage with this formulation is that it ensures that the algorithm is perfectly randomness-recoverable, even when it aborts.
\end{remark}

\paragraph{Random tapes notations.}
We describe the notations we use in this section when discussing \Cref{algo:isis-solver-mod2-detailed}'s internal randomness.
First note that, at any level $\ell > 1$, \Cref{algo:isis-solver-mod2-simple} constructs $m' = w^{\ell-1}$ blocks
For each of them but the last one, $n$ bits of randomness are consumed to sample the shares $\yv_i$'s.
And for each of them, $w - n - 1$ bits of randomness are consumed to sample the solutions.
In total, $(m'-1)n + m'(w - n - 1) = m'w - m' - n$ are consumed at each level $\ell > 1$, and $n$ are consumed at level $1$.
We then describe the random tape associated to any level $\ell > 1$ a set of vectors $\left(\{\yv_i\}_{i \in \Iint{1}{m'-1}}, \{\uv_i\}_{i \in \Iint{1}{m'}}\right)$, where $\yv_i \in \ZZ_2^{n}$, and $\uv_i \in \ZZ_2^{w - n - 1}$, and the one associated to level $1$ as a vector $\uv \in \ZZ_2^{w - n}$.
We denote the whole random tape (for all levels) $\rdtapef$, and often write $\rdtapef = \rdtape \Vert \rdtape'$, where $\rdtape$ denotes the random tape for the current level, and $\rdtape'$ denotes the one for the next levels.
Finally, we write $\isissolver(\Am, \yv; \rdtapef)$ to denote the execution of \Cref{algo:isis-solver-mod2-detailed} with random tape $\rdtapef$.

\paragraph{Randomness-recoverability of \Cref{algo:isis-solver-mod2-detailed}.}
In the following, we state propositions to show that \Cref{algo:isis-solver-mod2-detailed} is solution-uniform, and randomness-recoverable.
We start by showing that any two different random tapes lead to two different solutions (\Cref{prop:injectivity-isis-solver-mod2}) and that the algorithm is full-support (\Cref{prop:full-support-mod2}) in the sense that it can output all possible solutions, which directly implies solution-uniformity.
Finally, we prove that our algorithm is randomness-recoverable (\Cref{prop:rdrecov-mod2}).

\begin{proposition}
  \label{prop:injectivity-isis-solver-mod2}
  Let $n$, $\ell$, and $w$ be positive integers such that $w > cn$ for some constant $c > 1$, and let $q = 2^\ell$ and $m = w^\ell$.
  Let $\Am \in \ZZ_q^{n \times m}$, and $\yv \in \ZZ_q^n$.
  For all pairs $(\rdtapef_1, \rdtapef_2)$ of random tapes for $\isissolver$, if $\rdtapef_1 \neq \rdtapef_2$, then $\isissolver(\Am, y; \rdtapef_1) \neq \isissolver(\Am, y; \rdtapef_2)$.
\end{proposition}

\begin{proof}
  First note that, if $\isissolver$ aborts on $\rdtapef_1$ (resp. $\rdtapef_2)$, then the proof is immediate, as $\isissolver$ then returns $(\bot, \rdtapef_1)$, (resp. $(\bot, \rdtapef_2)$).
  We then focus on the case where none of the random tapes make the algorithm abort, and proceed by induction on $\ell$.
  When $\ell = 1$, the matrix $\tAm$ is square and of full-rank, which directly yields a one-to-one mapping between random tapes and solutions, and concludes the proof at this level.

  When $\ell > 1$, let $\xf_{1} \gets \isissolver(\Am, \yv; \rdtapef_1)$, and $\xf_{2} \gets \isissolver(\Am, \yv; \rdtapef_2)$ for $\rdtapef_1 \neq \rdtapef_2$.
  We prove in the following that $\xf_1 \neq \xf_2$.
  Write $\rdtapef_1$ as $(\rdtape_1 \Vert \rdtape'_1)$ and $\rdtapef_2$ as $(\rdtape_2 \Vert \rdtape'_2)$, and distinguish two cases:
  \begin{itemize}
    \item Assume first that $\rdtape_1 = \rdtape_2$, and $\rdtape'_1 \neq \rdtape'_2$.
    Let $(\xv_1, \Zm_1)$ be the pair of inhomogeneous and homogeneous solutions produced by $\isissolver(\Am, \yv; \rdtapef_1)$ at level-$\ell$, and $\xv'_1 \eqdef \isissolver(\Am', \yv'; \rdtape'_1)$.
    Similarly, let $(\xv_2, \Zm_2)$ be the pair of inhomogeneous and homogeneous solutions produced by $\isissolver(\Am, \yv; \rdtapef_2)$ at level-$\ell$, and $\xv'_2 \eqdef \isissolver(\Am', \yv'; \rdtape'_2)$.
    Then, by induction hypothesis, we have $\xv'_1 - \xv'_2$, and by construction, we have $(\xv_1, \Zm_1) = (\xv_2, \Zm_2)$.
    Thus, $\xf_1 - \xf_2 = \Zm_1(\xv'_1 - \xv'_2)$.
    As the kernel of $\Zm_1$ is $\{\zerov\}$, $\xf_1 \neq \xf_2$.

    \item Assume now that $\rdtape_1 \neq \rdtape_2$, and let $i$ be an block index such that $\rdtape_{1, i} \neq \rdtape_{2, i}$.
    Then, by construction, we have $\xv_{1, i} \neq \xv_{2, i}$, and $\zv_{1, i} \neq \zv_{2, i}$.
    Recall that $\xf_{1, i} = \xv_{1, i} + \zv_{1, i} \xv'_1[i]$ is a solution to the equation $\tAm_i \xv = (\yv_{1, i} \Vert \uv_{1, i})\T \mod 2$, and $\xf_{2, i} = \xv_{2, i} + \zv_{2, i} \xv'_2[i]$ is a solution to the equation $\tAm_i \xv = (\yv_{2, i} \Vert \uv_{2, i})\T \mod 2$.
    As $(\yv_{1, i} \Vert \uv_{1, i}) \neq (\yv_{2, i} \Vert \uv_{2, i})$, it directly comes $\xf_1 \neq \xf_2$, which concludes the proof.
  \end{itemize}
\end{proof}

\begin{remark}
  Note that this proposition is the reason why we cannot sample each $\xv_i$ by extending the block-matrix $\Am_i$ to be a square matrix, hence fixing a single solution, and letting $\zv_i$ be the difference between the two solutions produced by fixing the first $w-n-1$ bits of randomness, and letting the last one free.
  Indeed, this would imply that two random tapes for a level $\ell$ that differ only on the last bit correspond to the same $\zv_i$, and the second case ``$\rdtape_1 \neq \rdtape_2$'' would not carry over.
\end{remark}

The following proposition states that our $\ISIS$ solver is full-support with probability almost $1$.
\begin{proposition}
  \label{prop:full-support-mod2}
  Let $n$, $\ell$, and $w$ be positive integers such that $w > cn$ for some constant $c > 1$, and let $q = 2^\ell$ and $m = w^\ell$.
  Let $\Am \in \ZZ_q^{n \times m}$, and $\yv \in \ZZ_q^{n}$.
  Then, for all $\xf \in \Z^m_q$ such that $\Am\xf = \yv \mod q$, and $\xf \in \ZZ_2^m$.
  $$
    \prob{\xf \in \supp{\isissolver(\Am, \yv)}} = 1 - \negl(n)
  $$
  where the probability is taken over $\Am \sample \ZZ_q^{n \times m}$.
\end{proposition}

\begin{proof}
  We proceed by induction on $\ell$.
  When $\ell = 1$, let $\xf \in \Z^m_2$ such that $\Am\xf = \yv \mod 2$.
  Let $\uv \in \ZZ_2^{w-n}$ such that $(\yv \Vert \uv)\T = \widetilde{\Am} \xf \mod 2$, where $\widetilde{\Am}$ is the matrix obtained by running the matrix extension $\matext$ on input $\Am$.
  Note that $\tAm$ is a square matrix, so $\xf$ is the only solution of this equation.
  Thus, running $\isissolver$ on $(\Am, \yv)$ with random tape $\uv$ either yields $\bot$ if $\Am$ is not full-rank, or yields $\xf$ otherwise.
  The former case only happens with negligible probability, which concludes the proof at this level.

  Consider now a general level $\ell > 1$.
  Let $\xf \in \Z^m_q$ such that $\Am\xf = \yv \mod q$ and $\xf \in \ZZ_2^m$.
  We want to show that $\xf$ can be output by our $\ISIS$ solver algorithm, that is that $\xf$ can be written $\xf = \xv + \Zm \xv'$ where $(\xv, \Zm)$ is a pair of inhomogeneous-homogeneous solutions produced by the algorithm at this level, and $\xv'$ is an output of the solver at the next level $\ell - 1$.

  For all $i \in \{1, \dots, m'\}$, let $(\yv_i \| \uv_i)\T \eqdef \tAm_i \xf_i \mod 2$, where $\tAm_i$ is the extended block-matrix produced by $\matext$.
  We construct $\xv$ and $\Zm$ as in algorithm $\isissolver$: for all $i$, let $\xv_i \eqdef \xv^{(1)}_i$, and $\zv_i \eqdef \xv^{(2)}_i - \xv^{(1)}_i$, where $(\xv^{(1)}_i, \xv^{(2)}_i)$ is the pair, sorted by lexicographical order, of solutions to the equation $\tAm_i \xf_i = (\yv_i \| \uv_i)\T \mod 2$ --- note that $\xf_i \in \{\xv_i^{(1)}, \xv_i^{(2)}\}$ --- and let $\Zm$ be the corresponding block-columns matrix as defined in algorithm $\isissolver$.
  We then construct $\xv'$ as follows: for all $i$,
  $$
    \xv'[i] \eqdef \left\{\begin{array}{ll}
      0 \text{ if } \xv_i^{(1)} = \xv_i\\
      1 \text{ otherwise}
    \end{array}\right.
  $$
  It directly comes $\xf = \Zm\xv' \mod q$.
  Hence, by setting $\Am' \eqdef \Am \Zm / 2 \mod q'$, and $\yv' \eqdef (\yv - \Am \xv) / 2$ as in algorithm $\isissolver$, we have $\Am' \xv' = \yv' \mod q'$.
  By induction hypothesis, this implies that, with probability $1 - \negl(n)$, there exists a random tape $\rdtape'$ such that $\xv' = \isissolver(\Am', \yv'; \rdtape)$.
  Thus, by setting $\rdtape \eqdef ((\yv_i \| \uv_i)\T)_i$, and $\rdtapef \eqdef \rdtape \| \rdtape'$, it comes $\xf = \isissolver(\Am, \yv; \rdtapef)$, which concludes the proof.
\end{proof}

We can now conclude that our solver is almost solution-uniform.
\begin{proposition}
  \label{prop:sol-uniform-mod2}
  Let $n$, $\ell$, and $w$ be positive integers such that $w > cn$ for some constant $c > 1$, and let $q = 2^\ell$ and $m = w^\ell$.
  Let $\yv \in \ZZ_q^n$.
  Then algorithm $\isissolver$ is $\negl(n)$-close to be solution-uniform.
\end{proposition}

\begin{proof}
  The proof immediately follows from \cref{prop:injectivity-isis-solver-mod2,prop:full-support-mod2}.
\end{proof}

\begin{proposition}
  \label{prop:rdrecov-mod2}
  Let $n$, $\ell$, and $w$ be positive integers such that $w > cn$ for some constant $c > 1$, and let $q = 2^\ell$ and $m = w^\ell$.
  Then, for all $\Am \in \ZZ_q^{n \times m}$, \Cref{algo:isis-solver-mod2-detailed} is perfectly randomness-recoverable (\Cref{def:rdrecov}).
\end{proposition}

\begin{proof}
  We show that, given $\Am$, $\yv$, and any outcome of \Cref{algo:isis-solver-mod2-detailed}, we can efficiently recover the random tape $\rdtapef$ such that running \Cref{algo:isis-solver-mod2-detailed} on input $(\Am, \yv)$ with random tape $\rdtapef$ returns this outcome.
  With distinguish the following cases:
  \begin{enumerate}
    \item If the outcome is $(\bot, \rdtapef)$, we simply return $\rdtapef$.
    \item Else, if $\ell = 1$, the extended matrix $\tAm \in \ZZ_2^{w \times w}$ produced by $\matext$ procedure is a full-rank square matrix, hence there is a one-to-one mapping between the random tape $\rdtape = \uv$ and $\xf$.
    In this case, we then simply return $\uv \eqdef \tAm \xf$ as the random tape.

    \item Else, if $\ell > 1$, the outcome is a solution $\xf \in \ZZ_2^{n \times m}$, and, by \Cref{prop:injectivity-isis-solver-mod2}, there is a single way to write $\xf$ as $\xv + \Zm \xv' \mod q$.
    We then recover the random tape $\rdtape$ for this level by letting $(\yv_i \Vert \uv_i)\T \eqdef \tAm_i \xf_i \mod 2$, where $\tAm_i$ is the extended matrix obtained when running $\matext$ on $\Am_i$, and then $\rdtape \eqdef \{(\yv)_{i \in \Iint{1}{m'-1}}, (\uv)_{i \in \Iint{1}{m'}}\}$.
    Then we follow steps \ref{step:sol-sample} and \ref{step:merge-sols} of \Cref{algo:isis-solver-mod2-detailed} to construct solution pairs $(\xv_i, \zv_i)$ for all $i \in \Iint{1}{m'}$, and then construct $\xv$ and $\Zm$.
    Finally, we construct $\xv' \in \ZZ_2^{m'}$ by setting its $i$-th coefficient to be $\xv'[i] = 0$ if $\xf_i = \xv_i$, and $\xv'[i] = 1$ otherwise, and run this randomness-recovering procedure on $(\Am', \yv', \xv')$ to get the random-tape $\rdtape'$ of the next levels.
  \end{enumerate}
  By construction, in all cases, running \Cref{algo:isis-solver-mod2-detailed} on the recovered random tape yield the expected outcome.
\end{proof}

\paragraph{Putting everything together.}
By \Cref{prop:rdrecov-mod2,prop:sol-uniform-mod2,prop:time-complexity-mod2}, \Cref{algo:isis-solver-mod2-detailed} is perfectly randomness-recoverable, $\negl(n)$-solution-uniform, and runs in time $\poly\left(\log(n)\log(q)\right)$.
Then, by \Cref{th:iclwe-red-to-isis,Theorem:2}, we recover \Cref{Proposition:BJK+25}.

\paragraph{Extending to more general moduli.}
The correctness of \Cref{algo:isis-solver-mod2-simple} naturally extends to general moduli.
However, our strategy to make it full-support and randomness recoverable fails when considering moduli that are not powers of two.
The reason is that the crux of our proof lies in the fact that at every step of the induction, the partial solution $\xf = \xv + \Zm \xv' \mod q$ is small, \ie{} belongs to $\Z_2^m$.
To be able to keep it small, we need to find for every block matrix, a vector $\zv_i \in \ZZ_q^m$ such that both (1) $\Am_i \zv_i = 0 \mod p$ and (2) the solutions of $\Am_i \xv = \yv_i \mod p$ form a line $\{\xv_i^{(1)} + \alpha \zv_i \mod q \setmid \alpha \in \ZZ_p\}$, where $\xv_i^{(1)}$ is a fixed solution of this system (e.g. the smallest one, lexicographically-wise).

This is easily doable when $p = 2$, as there are only two such solutions, hence they always form a line.
This is however not the case for a general modulo, even another power-of-prime where the prime is not $2$.
Thus, this strategy fails and it is still unclear to us how to generalize this algorithm with the full-support and randomness recoverability properties for such moduli.
Interestingly, as discussed in \cite[Section 4.3]{BJK+25}, the reason why their algorithm works only for power-of-two moduli is the same as ours.

\paragraph{Comparison with \cite{BJK+25}.}

The algorithm of \cite{BJK+25} starts with many independent states ($\EDCPb$ samples), then applies a Kuperberg-like approach, consisting of recursively applying pairing-and-sieving operations which involve partial measurements to gradually build a ``target state'': a state whose Fourier amplitude concentrates on the secret $s$.
This process is iterative, adaptive, and requires maintaining a large collection of quantum states.

In contrast, our algorithm directly manipulates a single structured quantum state. By applying a QFT, we obtain in one step a state whose support lies on an $\ISIS$-type constraint.
Coherently applying the $\ISIS$ solver allows us to recover the same type of target state.
Our approach directly exploits the algebraic structure of the problem, allowing us to construct this target state in a single global step, and eliminating the need for explicit iterative pairing, sieving, or intermediate measurements.

Another advantage of our approach is that the quantum part of our algorithm is reduced to standard operations (mostly QFT) while the combinatorial complexity is delegated to the $\ISIS$ solver.
This yields a cleaner and more modular decomposition and highlights the connection with the lattice structure.

Regarding quantum resources however, our approach should not be interpreted as better than the one of \cite{BJK+25}: the overall complexity (time and memory) of \cite{BJK+25} matches ours when using the best known $\ISIS$ solver in this regime.
In particular, the large memory usage is not eliminated, but rather captured within the $\ISIS$ solver.

\bibliography{paper}

@string{it = {IEEE Trans. Inform. Theory}}

@article{KOW25,
  title={No exponential quantum speedup for $\mathrm{SIS}^\infty$ anymore},
  author={Kothari, Robin and O'Donnell, Ryan and Wu, Kewen},
  journal={arXiv preprint arXiv:2510.07515},
  year={2025}
}

@article{II24,
  title={Efficient quantum algorithms for some instances of the semidirect discrete logarithm problem},
  author={Imran, Muhammad and Ivanyos, G{\'a}bor},
  journal={Designs, Codes and Cryptography},
  volume={92},
  number={10},
  pages={2825--2843},
  year={2024},
  publisher={Springer}
}

@InProceedings{BN18,
	author="Bonnetain, Xavier
	and {Naya-Plasencia}, Mar{\'i}a",
	editor="Peyrin, Thomas
	and Galbraith, Steven",
	title="Hidden Shift Quantum Cryptanalysis and Implications",
	booktitle="Advances in Cryptology -- ASIACRYPT 2018",
	year="2018",
	publisher="Springer International Publishing",
	address="Cham",
	pages="560--592",
	isbn="978-3-030-03326-2"
}

@ARTICLE{FvdG99,
	author={Fuchs, C.A. and van de Graaf, J.},
	journal={IEEE Transactions on Information Theory}, 
	title={Cryptographic distinguishability measures for quantum-mechanical states}, 
	year={1999},
	volume={45},
	number={4},
	pages={1216-1227},
	doi={10.1109/18.761271}}

@InProceedings{Kup13,
	author =	{Kuperberg, Greg},
	title =	{{Another Subexponential-time Quantum Algorithm for the Dihedral Hidden Subgroup Problem}},
	booktitle =	{8th Conference on the Theory of Quantum Computation, Communication and Cryptography (TQC 2013)},
	pages =	{20--34},
	series =	{Leibniz International Proceedings in Informatics (LIPIcs)},
	ISBN =	{978-3-939897-55-2},
	ISSN =	{1868-8969},
	year =	{2013},
	volume =	{22},
	editor =	{Severini, Simone and Brandao, Fernando},
	publisher =	{Schloss Dagstuhl -- Leibniz-Zentrum f{\"u}r Informatik},
	address =	{Dagstuhl, Germany},
	URL =		{https://drops.dagstuhl.de/entities/document/10.4230/LIPIcs.TQC.2013.20},
	URN =		{urn:nbn:de:0030-drops-43213},
	doi =		{10.4230/LIPIcs.TQC.2013.20}
}

@InProceedings{SSTX09,
	author="Stehl{\'e}, Damien
	and Steinfeld, Ron
	and Tanaka, Keisuke
	and Xagawa, Keita",
	editor="Matsui, Mitsuru",
	title="Efficient Public Key Encryption Based on Ideal Lattices",
	booktitle="Advances in Cryptology -- ASIACRYPT 2009",
	year="2009",
	publisher="Springer Berlin Heidelberg",
	address="Berlin, Heidelberg",
	pages="617--635",
	isbn="978-3-642-10366-7"
}

@InProceedings{CHL+25,
	author="Chen, Yilei
	and Hu, Zihan
	and Liu, Qipeng
	and Luo, Han
	and Tu, Yaxin",
	editor="Tauman Kalai, Yael
	and Kamara, Seny F.",
	title="LWE with Quantum Amplitudes: Algorithm, Hardness, and Oblivious Sampling",
	booktitle="Advances in Cryptology -- CRYPTO 2025",
	year="2025",
	publisher="Springer Nature Switzerland",
	address="Cham",
	pages="513--544",
	isbn="978-3-032-01878-6"
}

@InProceedings{BJK+25,
	author="Bai, Shi
	and Jangir, Hansraj
	and Kirshanova, Elena
	and Ngo, Tran
	and Youmans, William",
	editor="Tauman Kalai, Yael
	and Kamara, Seny F.",
	title="A Quasi-polynomial Time Algorithm for the Extrapolated Dihedral Coset Problem over Power-of-Two Moduli",
	booktitle="Advances in Cryptology -- CRYPTO 2025",
	year="2025",
	publisher="Springer Nature Switzerland",
	address="Cham",
	pages="416--448",
	isbn="978-3-032-01878-6"
}

@inproceedings{BKSW18,
  title={Learning with errors and extrapolated dihedral cosets},
  author={Brakerski, Zvika and Kirshanova, Elena and Stehl{\'e}, Damien and Wen, Weiqiang},
  booktitle={IACR international workshop on public key cryptography},
  pages={702--727},
  year={2018},
  organization={Springer}
}

@inproceedings{Reg05,
		author = {Regev, Oded},
		title = {On lattices, learning with errors, random linear codes, and cryptography},
		year = {2005},
		isbn = {1581139608},
		publisher = {Association for Computing Machinery},
		address = {New York, NY, USA},
		url = {https://doi.org/10.1145/1060590.1060603},
		doi = {10.1145/1060590.1060603},
		booktitle = {Proceedings of the Thirty-Seventh Annual ACM Symposium on Theory of Computing},
		pages = {84–93},
		numpages = {10},
		location = {Baltimore, MD, USA},
		series = {STOC '05}
	}

@inproceedings{CT25,
		author = {Chailloux, Andr\'{e} and Tillich, Jean-Pierre},
		title = {Quantum Advantage from Soft Decoders},
		year = {2025},
		isbn = {9798400715105},
		publisher = {Association for Computing Machinery},
		address = {New York, NY, USA},
		url = {https://doi.org/10.1145/3717823.3718319},
		doi = {10.1145/3717823.3718319},
		booktitle = {Proceedings of the 57th Annual ACM Symposium on Theory of Computing},
		pages = {738–749},
		numpages = {12},
		location = {Prague, Czechia},
		series = {STOC '25}
	}

@article{YZ24,
  author       = {Takashi Yamakawa and
                  Mark Zhandry},
  title        = {Verifiable Quantum Advantage without Structure},
  journal      = {J. {ACM}},
  volume       = {71},
  number       = {3},
  pages        = {20},
  year         = {2024},
  url          = {https://doi.org/10.1145/3658665},
  doi          = {10.1145/3658665},
  bibsource    = {dblp computer science bibliography, https://dblp.org}
}

@article{DRT23,
  Author    = {{Debris-Alazard}, Thomas and
               Remaud, Maxime and
               Tillich, Jean{-}Pierre },
  title     = {Quantum Reduction of Finding Short Code Vectors to the Decoding Problem},
  journal   = it,
  Month = nov,
  Note = {in press, see also arXiv:2106.02747 (v2)},
  year      = {2023},
  doi       = {10.1109/TIT.2023.3327759},
  url       = {https://ieeexplore.ieee.org/document/10296874}}

@inproceedings{CLZ22,
  title={Quantum algorithms for variants of average-case lattice problems via filtering},
  author={Chen, Yilei and Liu, Qipeng and Zhandry, Mark},
  booktitle={Annual international conference on the theory and applications of cryptographic techniques},
  pages={372--401},
  year={2022},
  organization={Springer}
}

@inproceedings{DFS24,
	author       = {Thomas Debris{-}Alazard and
	Pouria Fallahpour and
	Damien Stehl{\'{e}}},
	editor       = {Bojan Mohar and
	Igor Shinkar and
	Ryan O'Donnell},
	title        = {Quantum Oblivious {LWE} Sampling and Insecurity of Standard Model
	Lattice-Based SNARKs},
	booktitle    = {Proceedings of the 56th Annual {ACM} Symposium on Theory of Computing,
	{STOC} 2024, Vancouver, BC, Canada, June 24-28, 2024},
	pages        = {423--434},
	publisher    = {{ACM}},
	year         = {2024},
	url          = {https://doi.org/10.1145/3618260.3649766},
	doi          = {10.1145/3618260.3649766},
	bibsource    = {dblp computer science bibliography, https://dblp.org}
}

@inproceedings{CT24,
	author       = {Andr{\'{e}} Chailloux and
	Jean{-}Pierre Tillich},
	title        = {The Quantum Decoding Problem},
	booktitle    = {Theory of Quantum Computation, Communication
	and Cryptography, {TQC} 2024, September 9-13, 2024, Okinawa, Japan},
	series       = {LIPIcs},
	volume       = {310},
	pages        = {6:1--6:14},
	year         = {2024}
}

@misc{JSW+24,
	title={Optimization by Decoded Quantum Interferometry}, 
	author={Stephen P. Jordan and Noah Shutty and Mary Wootters and Adam Zalcman and Alexander Schmidhuber and Robbie King and Sergei V. Isakov and Ryan Babbush},
	year={2024},
	eprint={2408.08292},
	archivePrefix={arXiv},
	primaryClass={quant-ph},
	url={https://arxiv.org/abs/2408.08292}, 
}

@Article{BKMH97,
  author = 	 {Ban, Masahi and Kurokawa, Keiko and Momose, Rei and Hirota, Osamu},
  title = 	 {Optimum measurements for discrimination among symmetric quantum states and parameter estimation},
  journal = 	 {International Journal of Theoretical Physics},
  year = 	 1997,
  volume = 	 36,
  number = 	 6,
  pages = 	 {1269--1288},
  doi      =     {10.1007/BF02435921},
  url      =     {https://doi.org/10.1007/BF02435921}}

@article{chailloux2025opi,
  title={OPI x Soft Decoders},
  author={Chailloux, Andr{\'e}},
  journal={arXiv preprint arXiv:2511.22691},
  year={2025}
}
\bibliographystyle{alpha}

\newpage

\begin{appendix}
\section{Tractability of $\SLWE$}\label{Appendix:Tractability}
We essentially reproduce (with elements in $\Z_q$ instead of $\F_q$), the argument of~\cite{CT24}. We consider the set of states $\{\ket{\widehat{\psi_\sv}}\}_{\sv \in \Z_q^n}$ with $\ket{\psi_\sv} = \sum_{\ev \in \Z_q^m} f(\ev) \ket{\Am\T \sv + \ev}$. The associated \emph{Pretty Good Measurement} is the POVM $\{M_\sv\}$ with 
$$ M_\sv = \rho^{-1/2} \kb{\widehat{\psi_\sv}}\rho^{-1/2}, \quad \text{with } \rho = \sum_{\sv} \kb{\widehat{\psi_\sv}}.$$
From \Cref{Proposition:PsiToW}, we have that 
$\ket{\psi_\sv} = \frac{1}{q^n} \sum_{\yv \in \Z_q^n} \omega^{\yv \cdot \sv}\sqrt{w_\yv} \ket{W_\yv},$
from which we can write 
\begin{align*}
	\kb{\psi_\sv} & = \frac{1}{q^{2n}} \sum_{\yv,\yv' \in \Z_q^n} \sqrt{w_\yv w_{\yv'}} \omega^{\sv \cdot(\yv - \yv')} \ketbra{W_{\yv}}{W_{\yv'}} 
\end{align*}
and 
\begin{align*}
	\rho & = \sum_{\sv} \kb{\psi_\sv} \\
	& = \frac{1}{q^{2n}} \sum_{\yv,\yv' \in \Z_q^n} \sqrt{w_\yv w_{\yv'}} \sum_{\sv} \omega^{\sv \cdot(\yv - \yv')}\ketbra{W_{\yv}}{W_{\yv'}} \\
	& = \frac{1}{q^n} \sum_{\yv \in \Z_q^n} w_\yv \kb{W_\yv}
\end{align*}
Since the $\ket{W_\yv}$ are pairwise orthogonal and of norm $1$, we have $\rho^{-1/2} = \sum_{\yv \in \Z_q^n}\sqrt{\frac{q^n}{w_\yv}} \kb{W_\yv}$. Now, we write 
\begin{align*}
	\rho^{-1/2} \cdot \ket{\widehat{\psi_\sv}} & =  \left(\sum_{\yv \in \Z_q^n}\sqrt{\frac{q^n}{w_\yv}} \kb{W_\yv}\right) \cdot \left(\frac{1}{q^n} \sum_{\yv \in \Z_q^n} \omega^{\yv \cdot \sv}\sqrt{w_\yv} \ket{W_\yv}\right) \\
	& = \frac{1}{\sqrt{q^n}} \sum_{\yv \in \Z_q^n} \omega^{\yv \cdot \sv}\ket{W_\yv} \eqdef \ket{Y_\sv}
\end{align*}
where $\ket{Y_\sv}$ is a pure unit vector. By definition, we then have that $M_\sv = \kb{Y_\sv}$. The probability that the Pretty Good Measurement succeeds is then \begin{align*}
	p_{PGM} = E_{\sv} |\braket{\psi_\sv}{Y_\sv}|^2 = \frac{1}{q^{3n}} \left(\sum_{\yv \in \Z_q^n} \sqrt{w_\yv}\right)^2 = \left(\E_{\yv \in \Z_q^n} \left[\sqrt{\frac{w_\yv}{q^n}}\right]\right)^2
\end{align*}

Finally, we know that this measurement is optimal for distinguishing between the states $\ket{\psi_\sv}$ due to the symmetric nature of this set of states \cite{BKMH97}.

\section{Instantiation of our forward theorem with Gaussian distributions}\label{Appendix:Gaussians}
We directly define the discrete Gaussian probability function.
\begin{definition}
	Let $r > 0$. We define the Gaussian probability function of parameter $\sigma$ as follows
	$$ \rho_r(x) = \frac{e^{-\frac{\pi x^2}{r^2}}}{\sum_{y \in \Z_q} e^{-\frac{\pi y^2}{r^2}}}.$$
\end{definition}

The discrete Gaussian distribution has 

\begin{proposition}
	Let $\sigma > 0$ and let $f = \sqrt{\rho_r(x)}$. We have $\hf =  \sqrt{\rho_{r^*}} + o(1)$, with $r^* = \frac{q}{2r}$. Moreover, for $T = \{\xv \in \Z_q^m : \norm{\xv}_2 \le \sqrt{m} \frac{q}{\sqrt{8\pi}r}\}$, we have $\sum_{\yv \in T} |\hf(\yv)|^2 = \Theta(1)$.
\end{proposition}

\begin{proof}
	First notice that $f$ is proportional to $\rho_{\sqrt{2}r}$. The Poisson formula in the continuous setting tells us that the Fourier transform of a Gaussian of parameter $r$ is a Gaussian of parameter $\frac{q}{r}$. This also holds approximately for the discrete Gaussian (see for instance~\cite{DFS24}, Lemma 4) which tells us $\hf = \rho_{\frac{q}{\sqrt{2}r}} + o(1) = \sqrt{\rho_{\frac{q}{2r}}} + o(1)$. For the second part, we just use that a vector in $\Z_q^m$ where each coordinate is chosen according to a Gaussian distribution of parameter $\frac{q}{2r}$ has typical $2$-norm $\sqrt{m} \frac{q}{\sqrt{8\pi}r}$.
\end{proof}

As a direct corollary, using this in Theorem~\ref{Theorem:1}, we obtain the following statement 

\begin{proposition}
	Let $r > 0$. If we have an efficient algorithm for $\SLWE(\Am,\sqrt{\rho_r})$, or if we have an efficient algorithm for $\LWE(\Am,\rho_r)$, then we have an efficient quantum algorithm for $\ISIS_2(\Am,\sqrt{\frac{qm}{8\pi r}})$. 
\end{proposition}

\section{Proof of reverse reduction with almost randomness recoverable algorithms for $\ISIS$}
\label{Appendix:ImperfectRR}
In this section, we extend Theorem~\ref{Theorem:3} to the theorem below
\begin{theorem}
	Let $\aa$ be an algorithm for $\ISIS(\Am, T)$ with time complexity $t$.
	Assume the following properties:
	\begin{itemize}
		\item $\aa$ is $\eps$-close to being solution-uniform {\ie} if we define $p_\yv(\xv) = \Pr_{\rv \Unif \zo^\ell}\left(\aa(\yv,\rv) = \xv\right)$, we have 
		$$ \Delta(p_\yv,u_\yv) = \eps_\yv \quad \text{and} \quad \E_{\yv \Unif \Z_q^n}\left[\eps_\yv\right] = \eps,$$
		where $\Delta(p_\yv,u_\yv)$ is the statistical distance between the distribution $p_\yv$ and the probability function $u_\yv = \frac{1}{|T \cap \Lambda^\bot_{\yv}(\Am)|} \one_{T \cap \Lambda^\bot_{\yv}(\Am)}$.
		\item $\aa$ is $\varepsilon'$-randomness-recoverable; let $\rdext$ denote a randomness-extractor for $\aa$ succeeding with probability $1 - \varepsilon'$.
		That is, for all $\yv \in \Z_q^n$,
		$$
		\probsublong{r \sample \bin^\ell}{\rdext(\yv,\aa(\yv; r)) = r} = 1 - \varepsilon'.
		$$
	\end{itemize}
	Then there exists an efficient algorithm solving $\ICLWE(\Am, f)$ with $\hf = \one_T$ with fidelity $1-\eps- \varepsilon'$ and time complexity $\poly(t)$.
\end{theorem}

\begin{remark}
	We do not explicitly say anything about the success probability of the algorithm $\aa$. However, the fact that it is $\eps$-close to being solution-uniform implies that with a uniformly random choice of randomness $r$, the algorithm outputs a valid solution ({\ie} $\in T \cap \Lambda^\bot_\yv(\Am)$) with probability at least $1-\eps$ on average on $\yv$.
\end{remark}

\begin{proof}
	We show that there is an efficient process mapping $\ket{\yv}\ket{0}$ to $\ket{\yv}\ket{W_\yv}$ for all $\yv \in \Z^n_q$.
	Start with $\ket{\yv}$ as first register, then prepare a uniform superposition of $r \in \bin^\ell$ in the second register:
	$$
	\frac{1}{\sqrt{2^\ell}}\ket{\yv} \sum_{r \in \bin^\ell}\ket{r}
	$$
	Apply then $\aa$ coherently over the first and second registers and store the result in a third register, then swap the second and third registers.
	The resulting state is
	$$
	\frac{1}{\sqrt{2^\ell}}\ket{\yv} \sum_{r \in \bin^\ell}\ket{\aa(\yv; r)}\ket{r}
	$$
	Apply now the randomness-extractor $\rdext$ coherently over the first and second registers, and substract the result from the second one; let $\ket{\psi_\yv}$ denote the resulting state:
	$$
	\ket{\psi_\yv} = \frac{1}{\sqrt{2^\ell}}\ket{\yv} \sum_{r \in \bin^\ell} \ket{\aa(\yv; r)} \ket{r - \rdext(\yv, \aa(\yv; r))}
	$$
	
	Consider now a state $\ket{\psi'_\yv}$ constructed in the same way with an ``ideal'' $\ISIS$-solver $\aa'$, that is a perfectly solution-uniform and perfectly randomness-recoverable --- with corresponding ideal randomness-extractor $\rdext'$ --- solver:
	$$
	\ket{\psi'_\yv} = \frac{1}{\sqrt{2^\ell}}\ket{\yv} \sum_{r \in \bin^\ell} \ket{\aa'(\yv; r)} \ket{0}
	$$
	where the last register is $\ket{0}$ because the randomness-extractor is perfect.
	We now discard this last register, and apply an inverse $\QFT$ to the second register.
	By \Cref{Proposition:WFourier}, this yields $\ket{\yv}\ket{W_\yv}$.
	
	It remains to show that, averaged over $\yv$, the overlap between $\ket{\psi_\yv}$ and $\ket{\psi'_\yv}$ is larger than $1 - \varepsilon - \varepsilon'$.
	We have
	\begin{align*}
		\left|\braket{\psi_\yv}{\psi'_\yv}\right|  &= \frac{1}{2^\ell}\ket{\yv} \sum_{r, r' \in \bin^\ell} \braket{\aa(\yv; r)}{\aa'(\yv; r')} \braket{r - \rdext(\yv, \aa(\yv; r))}{0}\\
		~                             &= \frac{1}{2^\ell}\ket{\yv} \sum_{r \in \bin^\ell} \left(\sum_{r' \in \bin^\ell}\braket{\aa(\yv; r)}{\aa'(\yv; r')}\right) \braket{r - \rdext(\yv, \aa(\yv; r))}{0}\\
	\end{align*}
	Now, because $\aa'$ is perfectly solution-uniform and perfectly randomness-recoverable, there is, for every $r$, at most a single $r'$ such that $\aa(\yv; r) = \aa'(\yv; r')$.
	Hence the first overlap $\sum_{r' \in \bin^\ell}\braket{\aa(\yv; r)}{\aa'(\yv; r')}$ is $1$ if $\aa(\yv; r) \in \Lambda^\perp_\yv \cap T$, and $0$ otherwise.
	Furthermore, the second overlap $\braket{r - \rdext(\yv, \aa(\yv; r))}{0}$ is $1$ if $r = \rdext(\yv, \aa(\yv; r))$, and $0$ otherwise.
	It comes
	\begin{align*}
		\left|\braket{\psi_\yv}{\psi'_\yv}\right| &= \probsub{r \sample \bin^\ell}{\aa(\yv; r) \in \Lambda^\perp_\yv \land \rdext(\aa(\yv; r)) = r}\\
		&= 1 - \probsub{r \sample \bin^\ell}{\aa(\yv; r) \not\in \Lambda^\perp_\yv \lor \rdext(\aa(\yv; r)) \neq r}
	\end{align*}
	
	By definition of statistical distance, we have $\probsub{r \sample \bin^\ell}{\aa(\yv; r) \not\in \Lambda^\perp_\yv} \leq \epsilon$, and $\probsub{r \sample \bin^\ell}{\rdext(\aa(\yv; r)) \neq r} = \epsilon'$ by definition of an $\epsilon'$-randomness-recoverable solver.
	Applying union bound thus gives us $\left|\braket{\psi_\yv}{\psi'_\yv}\right| \geq 1 - \varepsilon - \varepsilon'$ which concludes the proof.
\end{proof}

\section{Proof of Theorem~\ref{Theorem:2}}\label{Appendix:Reverse1}
	
This section will be devoted to the proof of Theorem~\ref{Theorem:2}. We first prove the following lemma.

\begin{lemma}\label{Lemma:Addy}
	The unitary $U$ that for each $\yv \in \Z_q^n$ satisfies $U : \ket{W_\yv}\ket{0} \rightarrow \ket{W_\yv}\ket{\yv}$ is efficiently computable. 
\end{lemma}
\begin{proof}
	Let $V_\Am$ be the quantum unitary satisfying $V_{\Am}\ket{\xv}\ket{0} = \ket{\xv}\ket{\Am\xv}$.	We perform the following operations, using the expressions of $\ket{\widehat{W_\yv}}$ from \Cref{Proposition:WFourier}:
	\begin{align*}
		\ket{W_\yv}\ket{0} \xrightarrow{QFT_{\Z_q^m} \otimes I} \frac{q^n}{\sqrt{w_\yv}} \sum_{\xv \in \Lambda^\bot_\yv(\Am)} \hf(\xv)\ket{\xv}\ket{0} 
		&\xrightarrow{V_\Am} \frac{q^n}{\sqrt{w_\yv}} \sum_{\xv \in \Lambda^\bot_\yv(\Am)} \hf(\xv)\ket{\xv}\ket{\yv} \\
		&\xrightarrow{QFT^\dagger_{\Z_q^m} \otimes I} \ket{W_\yv}\ket{\yv}.
	\end{align*}
\end{proof}
We can now prove our  theorem 
\begin{proof}
	We start from an efficient unitary $U$ for $\ICLWE(\Am,f)$.
	$$
	\forall \yv \in \Z_q^n, \ U\ket{\yv}\ket{\zerov} = \ket{\yv}\ket{\wphi_\yv} \quad \text{with } \E_{\yv \Unif \Z_q^n}\left[\left|\braket{\wphi_\yv}{W_\yv}\right|\right] = \gamma = 1 - \eps'.
	$$
	For each $\sv \in \Z_q^n$, we define 
	$$ \ket{B_\sv} = \frac{1}{\sqrt{q^n}}\sum_{\yv \in \Z_q^n} \omega^{\sv \cdot \yv} \ket{\wphi_\yv}\ket{\yv}.$$
	The $\ket{B_\sv}$ are pairwise orthogonal unit vectors. Using $U$, one can efficiently construct the unitary $\ket{\sv}\ket{0} \rightarrow \ket{B_\sv}$. Indeed, we can write
	\begin{align*}
		\ket{\sv}\ket{0} \xrightarrow{QFT \otimes I} \frac{1}{\sqrt{q^n}}\sum_{\yv \in \Z_q^n} \omega^{\yv \cdot \sv}\ket{\yv}\ket{0} \xrightarrow{U}  \frac{1}{\sqrt{q^n}}\sum_{\yv \in \Z_q^n} \omega^{\yv \cdot \sv}\ket{\yv}\ket{\wphi_\yv} \xrightarrow{SWAP} \ket{B_\sv}
	\end{align*}
	In particular, we can efficiently measure in the basis $\ket{B_\sv}$. We now present our algorithm for $\SLWE(\Am,f)$.
	\begin{enumerate}
		\item Start from $\ket{\psi_\sv} = \frac{1}{q^n} \sum_{\yv \in \Z_q^n} \omega^{\yv \cdot \sv}\sqrt{w_\yv} \ket{W_\yv}$ (see \Cref{Proposition:PsiToW}), and  apply the unitary from \Cref{Lemma:Addy} to obtain the state 
		$$ \ket{\psi'_\sv} = \frac{1}{q^n} \sum_{\yv \in \Z_q^n} \omega^{\yv \cdot \sv} \sqrt{w_\yv}\ket{W_\yv}\ket{\yv}.$$ 
		\item Measure this state in the basis $\{\ket{B_\sv}\}$ and output the result.
	\end{enumerate}
	The success probability of this algorithm is $\E_{\sv \Unif \Z_q^n} |\braket{\psi'_\sv}{B_\sv}|^2$. We compute 
	\begin{align*}
		|\braket{\psi'_\sv}{B_\sv}| = \frac{1}{q^n}\frac{1}{\sqrt{q^n}}\sum_{\yv \in \Z_q^n} \sqrt{w_\yv}\gamma_\yv = \E_{\yv}\left[\sqrt{\frac{w_\yv}{q^n}}\gamma_\yv\right].
	\end{align*}
	where $\gamma_\yv = |\braket{W_\yv}{\wphi_\yv}|$.
	In order to prove this statement, we define $\eps_y = 1 - \sqrt{\frac{w_\yv}{q^n}}$ and $\eps'_\yv = 1 - \gamma_\yv$. Let us now recap what we know about these variables.
	\begin{itemize}
		\item $\E_{\yv}\left[\sqrt{\frac{w_\yv}{q^n}}\right] = (1-\eps)$ hence $\E_{\yv}\left[\eps_\yv\right] = \eps$.
		\item $\E_{\yv}\left[{\frac{w_\yv}{q^n}}\right] = 1$ which can be rewritten $\E_{\yv}\left[(1 - \eps_\yv)^2\right] = 1$ which implies $\E_{\yv}[\eps_\yv^2] = 2\E_{\yv}[\eps_\yv] =2\eps$.
		\item $\E_{\yv}\left[\gamma_\yv\right] = \gamma$ so $\E_{\yv}\left[\eps'_\yv\right] = \eps'$
		\item Each $\gamma_\yv \in [0,1]$ so each $\eps'_\yv \in [0,1]$ and $\E_\yv\left[\eps'^2_\yv\right] \le \E_\yv\left[\eps'_\yv\right] = \eps'$.
	\end{itemize}
	We can now compute the success probability of this algorithm. We write 
	\begin{align*}
		|\braket{\psi'_\sv}{B_\sv}| & = \E_{\yv}\left[\sqrt{\frac{w_\yv}{q^n}}\gamma_\yv\right] = \E_\yv\left[(1-\eps_\yv)(1-\eps'_\yv)\right] = 1 - \eps - \eps' + \E_{\yv}\left[\eps_\yv \eps'_\yv\right]
	\end{align*}
	Now, using the Cauchy-Schwarz inequality, we obtain
	\begin{align*}
		\left|\E_{\yv}\left[\eps_\yv \eps'_\yv\right]\right| \le \sqrt{\E_\yv\left[\eps_\yv^2\right]} \sqrt{\E_\yv\left[\eps'^2_\yv\right]} \le \sqrt{2\eps\eps'}
	\end{align*}
	Plugging this in the above, we obtain that for each $\sv$, $|\braket{\psi'_\sv}{B_\sv}| \ge 1 - \eps - \eps' - 2\sqrt{\eps\eps'}$. Since the success probability is $\E_{\sv}\left[|\braket{\psi'_\sv}{B_\sv}|^2\right]$, we get the desired result. 
\end{proof}

\section{Matrice Extension Subroutine}
\label{app:matext}
We describe in this section an example of a deterministic matrix extension procedure to be used as the $\matext$ procedure in \Cref{algo:isis-solver-mod2-detailed}.
It takes a rank-$n$ matrix $\Am \in \ZZ_2^{n \times w}$ as input, as well as an integer $r \in \Iint{n}{w}$, and returns an ``extended'' matrix $\tAm \in \ZZ_2^{r \times w}$ of rank $r$, and such that the first $n$ rows of $\tAm$ exactly form the matrix $\Am$.

\begin{algorithm*}
  \caption{Extend a full rank matrix in $\mathbb{Z}_2^{n \times w}$ to a rank-$r$ matrix in $\Z_2^{r \times w }$ deterministically.}
  \label{proc:matrix-ext1}
  \textbf{Input:} A matrix $\Am \in \mathbb{Z}_2^{n \times w}$ of rank $n$, and an integer $r \in \Iint{n}{w}$.\\
  \textbf{Output:} A matrix $\widetilde{\Am} \in \mathbb{Z}_2^{r \times w}$ of rank $r$, and whose first $n$ rows are $\Am$.
  \begin{algorithmic}[1]
    \State Let $\ev_i^\top$ be the row vector in $\mathbb{Z}_2^{w}$ with a $1$ at position $i$ and $0$ everywhere else.
    \State $\widetilde{\Am} \gets \Am$
    \State $i \gets 0$
    \While{$\operatorname{rank}(\widetilde{\Am}) < r$}
    \State $\widetilde{\Am}' =
    \left[
    \begin{array}{c}
      \widetilde{\Am} \\
      \hline  \vspace{-0.33cm} \\
      \ev_i^\top
    \end{array}
    \right]$ \Comment{Append row $\ev_i\T$ to $\widetilde{\Am}$}
    \If{$\operatorname{rank}(\widetilde{\Am}') > \operatorname{rank}(\widetilde{\Am})$}
    \State $\widetilde{\Am} \gets \widetilde{\Am}'$
    \EndIf
    \State $i \gets i + 1$
    \EndWhile
    \State \Return $\widetilde{\Am}$
  \end{algorithmic}
\end{algorithm*}
 \end{appendix}

\end{document}